\documentclass[a4paper,UKenglish]{amsart}

\usepackage{microtype}

\DeclareSymbolFont{bbold}{U}{bbold}{m}{n}
\DeclareSymbolFontAlphabet{\mathbbold}{bbold}

\theoremstyle{plain} 
\newtheorem{theorem}{Theorem}
\newtheorem{thm}[theorem]{Theorem}
\newtheorem{lem}[theorem]{Lemma}
\newtheorem{cor}[theorem]{Corollary}
\newtheorem{ADC}[theorem]{Algebraic Dichotomy Conjecture (ADC)}

\newtheorem{AoNT}[theorem]{All or Nothing Theorem (ANT)}
\newtheorem{GDfSG}[theorem]{Gap Dichotomy for Simple Graphs}

\newtheorem*{remark}{Remark}

\theoremstyle{definition}
\newtheorem{defn}[theorem]{Definition}
\newtheorem{definition}[theorem]{Definition}

\newtheorem{rem}[theorem]{Remark}

\newcommand{\arity}{\operatorname{arity}}
\newcommand{\diag}{\operatorname{diag}}

\newcommand{\Pol}{\operatorname{Pol}}

\newcommand{\pp}{\operatorname{pp}}

\DeclareMathOperator{\Con}{Con}

\newcommand{\SAT}{\textup{SAT}}

\newcommand{\SEP}{\textup{SEP}}

\newcommand{\family}{\operatorname{family}}

\newcommand{\dsharp}{\raisebox{.01em}{\kern-.08em\scalebox{.6}{\rotatebox[origin=c]{45}{\ding{67}}}}\kern-.017em}
\newcommand{\type}{\operatorname{type}}

\DeclareMathOperator{\CSP}{CSP}

\newcommand{\up}[1]{\textup{#1}}

\newcommand{\rest}[1]{{\upharpoonright}_{#1}}

\newcommand{\Mod}{\texttt{Mod}}

\newcommand{\bigand}{\operatornamewithlimits{\hbox{\Large$\&$}}}
\newcommand{\Bigand}{\operatornamewithlimits{\hbox{\LARGE$\&$}}}

\title[A promise problem dichotomy for constraint problems]{All or nothing: toward a promise problem dichotomy for constraint problems\footnote{The second author was  supported by ARC Discovery Project DP1094578 and ARC Future Fellowship FT120100666.}}

\author[1]{Lucy Ham}
\author[2]{Marcel Jackson}
\address[1,2]{Department of Mathematics and Statistics, La Trobe University, Melbourne, Australia\\
  \texttt{leham@students.latrobe.edu.au}\\
  \texttt{m.g.jackson@latrobe.edu.au}}


\usepackage{amssymb,latexsym,amsmath,amsthm,enumerate}
\usepackage{xcolor,stmaryrd}
\usepackage[mathscr]{eucal}
\usepackage{mathrsfs,pifont}
\usepackage{framed}
\usepackage[all]{xy} 
\usepackage{tikz}
\usetikzlibrary{positioning}

\begin{document}

\maketitle

\begin{abstract}
A finite constraint language $\mathscr{R}$ is a finite set of relations over some finite domain $A$.  We show that intractability of the constraint satisfaction problem $\CSP(\mathscr{R})$ can, in all known cases, be replaced by an infinite hierarchy of intractable promise problems of increasingly disparate promise conditions: where instances are guaranteed to either have no solutions at all, or to be $k$-robustly satisfiable (for any fixed $k$), meaning that every ``reasonable'' partial instantiation on~$k$ variables extends to a solution.  For example, subject to the assumption $\texttt{P}\neq \texttt{NP}$, then for any~$k$, we show that there is no polynomial time algorithm that can distinguish non-$3$-colourable graphs, from those for which any reasonable $3$-colouring of any $k$ of the vertices can extend to a full $3$-colouring.  Our main result shows that an analogous statement holds for  all known intractable constraint problems over fixed finite constraint languages.
\end{abstract}


\section{Introduction}

In the constraint satisfaction problem (CSP) we are given a domain $A$, a list of relations $\mathscr{R}$ on $A$ and a finite set $V$ of variables, in which various tuples of variables have been constrained by the relations in $\mathscr{R}$.  The fundamental \emph{satisfaction} question is to decide whether there is a function $\phi:V\to A$ such that $(\phi(v_1),\dots,\phi(v_n))\in r$ whenever $\langle (v_1,\dots,v_n),r\rangle$ is a constraint (and $r\in\mathscr{R}$ is of arity $n$). 
Many important computational problems are expressible in this framework, even in the particular case where the domain $A$ and relations $\mathscr{R}$ are fixed.  
Such \emph{fixed template} CSPs have received particular attention in theoretical investigations, and are also the focus of the present article: examples include the SAT variants considered by Schaefer~\cite{sch}, graph homomorphism problems such as in the Hell-Ne\v{s}et\v{r}il dichotomy~\cite{helnes} as well as list-homomorphism problems and conservative CSPs \cite{bul:LH}.  
Feder and Vardi \cite{fedvar} generated particular attention on the theoretical analysis of computational complexity of fixed template CSPs, by tying the complexity of fixed finite template CSPs precisely to those complexities to be found in the largest logically definable class for which they were unable to prove that Ladner's Theorem holds.  This motivated their famous \emph{dichotomy conjecture}: is it true that a fixed finite template CSP is either solvable in polynomial time or is \texttt{NP}-complete?

A pivotal development in the efforts toward a possible proof of the dichotomy conjecture was the introduction of universal algebraic methods.  This provided fresh tools to build tractable algorithms, and to build reductions for hardness,  as well as an established  mathematical landscape in which to formulate conjectures on complexity.  The method is fundamental to Bulatov's classification of 3-element CSPs \cite{bul3}, of the Dichotomy Theorem for conservative CSPs \cite{bul:LH}, for homomorphism problems on digraphs without sources and sinks~\cite{BKN}, in the classification of when a CSP is solvable by generalised Gaussian elimination~\cite{IMMVW}, and of when a CSP is solvable by a local consistency check algorithm~\cite{barkoz:BW}, among others.  The \emph{algebraic dichotomy conjecture} (ADC) of \cite{BJK} refines the Feder-Vardi conjecture by speculating the precise boundary between \texttt{P} and \texttt{NP}, in terms of the presence of certain algebraic properties.  The ADC has been verified in each of the aforementioned tractability classifications.

The present article shows that \texttt{NP}-completeness results obtained via the algebraic method also imply the \texttt{NP}-completeness of a surprisingly strong promise problem.  The NO instances are those for which there is no solution, but the YES instances are instances for which any ``reasonable'' partial assignment on $k$ variables can extend to a solution.  
``Reasonable'' here means subject to some finite set of local, necessary conditions.  The main result---the All or Nothing Theorem (ANT)~\ref{thm:main}---proves the \texttt{NP}-completeness of this promise problem for any integer $k\geq 0$ and in any intractable CSP covered by the algebraic method.  The promise conditions include satisfaction as a special case, and complement the promise condition on NO instances provided by the PCP Theorem \cite{arosaf} (at least $\varepsilon$ proportion of the constraints must fail, for some $\varepsilon>0$). 
We are also able to prove a dichotomy theorem by showing that for sufficiently large $k$, our promise problem is solvable in $\texttt{AC}^0$ if and only if the CSP is of bounded width (in the sense of Barto and Kozik \cite{barkoz:BW}) and otherwise is hard for the complexity class $\texttt{Mod}_p(\texttt{L})$ for some prime~$p$.

A second contribution of the article is to connect the model-theoretic notion of \emph{quasivariety} to the concept of \emph{implied constraints}.  Identifying implied constraints is a central method employed in constraint solvers~\cite{MSS}, and the proliferation of implied constraints is associated with phase transitions in randomly generated constraint problems;~\cite{MZKST}.  We explain how the absence of implied constraints corresponds to membership in the quasivariety generated by the template.   Intuitively, it seems quite unlikely that the problem of recognising ``no implied constraints'' can be approached using the algebraic method, because there is no obvious reduction between constraint languages $\mathscr{R}_1$ and $\mathscr{R}_2$ when $\mathscr{R}_1\subsetneq \mathscr{R}_2$.  
Despite this intuition, the strength of the promise in the ANT enables us to show that whenever the algebraic method shows hardness of $\CSP(\mathscr{R})$, then there is no polynomial time algorithm to distinguish constraint instances with no solution, from those that have no implied constraints.   We can also use our bounded width dichotomy to obtain the most general  nonfinite axiomatisability result known for finitely generated quasivarieties.  A further important corollary is a promise problem extension of Hell and Ne\v{s}et\v{r}il's well-known dichotomy for simple graphs~\cite{helnes}.  The strength of the promise in the ANT~\ref{thm:main} enables us to show that, under extremely general conditions, there is no polynomial time algorithm to distinguish constraint instances with no solutions at all, from those for which there are no implied constraints at all.   The result also implies an equally general hardness result for deciding membership in finitely generated quasivarieties, as well as the most general known nonfinite axiomatisability result known for finitely generated quasivarieties.  A further important corollary is a promise problem extension of the Hell and Ne\v{s}et\v{r}il's well-known dichotomy for simple graphs \cite{helnes}.

More generally, the ability to extend all reasonable partial assignments appears to hold potential for further application.  Aside from similarity to the ultrahomogeneity concept from model theory, the concept has already found applications to minimal networks in~\cite{got}, in quantum mechanics~\cite{AGK}, and semigroup theory~\cite{jac:SAT}.  


%
\section{Constraints and implied constraints}\label{sec:CSP}
Since Feder and Vardi \cite{fedvar} it has been standard to reformulate the fixed template CSP over domain $A$ and language $\mathscr{R}$ as a homomorphism problem between model-theoretic structures. The template is a relational structure $\mathbb{A}=\langle A,\mathscr{R}^A\rangle$ (with $\mathscr{R}$ a relational signature), as is the instance $(V,\mathscr{C})$ (where $\mathscr{C}$ is the list of constraints) where the variable set $V$ is the universe, and with each $r\in\mathscr{R}$ being interpreted as the relation on $V$ equal to the set of tuples constrained to $r$ in the set $\mathscr{C}$.  Thus each individual constraint $\langle (v_1,\dots,v_n),r\rangle$ becomes a membership of a tuple $(v_1,\dots,v_n)$ in the relation $r^V$ on $V$.  
We refer to  $(v_1,\dots,v_n)\in r^V$ as a \emph{hyperedge}.   The \emph{constraint satisfaction problem for $\mathbb{A}$}, which we denote by $\CSP(\mathbb{A})$ is the problem of deciding membership in the class of finite structures admitting homomorphism into~$\mathbb{A}$.  Throughout the article, $\mathbb{A}$  will be the default notation for a CSP template of signature $\mathscr{R}$ (both assumed to be finite) and $\mathbb{B}$ for a 
general (finite) $\mathscr{R}$-structure.  We let $\arity(\mathscr{R})$ denote the maximal arity of any relation in $\mathscr{R}$. 

An \emph{implied constraint} in $\mathbb{B}$ relative to $\mathbb{A}$ is a nonhyperedge $(v_1,\dots,v_n)\notin r^B$, where $r\in \mathscr{R}$, that is mapped to an $r^A$ hyperedge under every homomorphism into~$\mathbb{A}$.  We also allow this for the logical relation $=$, thus an implied equality is a pair $b_1,b_2$ such that every homomorphism identifies $b_1$ with $b_2$.  Thus an input $\mathbb{B}$ has \emph{no implied constraints} if and only if every nonhyperedge can be homomorphically mapped to a nonhyperedge over the same relation in $\mathbb{A}$.  This ``separation condition'' is widely known to be equivalent to the property that $\mathbb{B}$ lies in the \emph{quasivariety} generated by $\mathbb{A}$: the class of isomorphic copies of induced substructures of direct powers of $\mathbb{A}$; see Maltsev \cite{mal} and Gorbunov \cite{gor}, but also \cite[Theorem~2.1]{jac:SAT} and \cite[\S2.1,2.2]{jactro} for the CSP interpretations and generalizations.  The trivial one-element structure ${\bf 1}_\mathscr{R}$ satisfies the separation condition  vacuously.  If we wish to exclude ${\bf 1}_\mathscr{R}$ we arrive at the \emph{universal Horn class} generated by $\mathbb{A}$ (which excludes the zeroth power from ``direct powers'').  We let $\mathsf{Q}(\mathbb{A})$ denote the quasivariety of $\mathbb{A}$ and $\mathsf{Q}^+(\mathbb{A})$ the universal Horn class of~$\mathbb{A}$.  Membership in $\mathsf{Q}(\mathbb{A})$ is the problem of deciding if an input has no implied constraints, which we denote by $\CSP_{\infty}(\mathbb{A})$.  Membership in $\mathsf{Q}^+(\mathbb{A})$ is essentially the same as $\CSP_{\infty}(\mathbb{A})$ because $\mathsf{Q}(\mathbb{A})$ and $\mathsf{Q}^+(\mathbb{A})$  differ on at most the structure ${\bf 1}_\mathscr{R}$.  
\smallskip

\fbox{
\parbox{0.92\textwidth}{
Problem: $\CSP_{\infty}(\mathbb{A})$ (no implied constraints)\\
Instance: a finite $\mathscr{R}$-structure $\mathbb{B}$.\\
Question: for every nonhyperedge $(v_1,\dots,v_n)\notin r^B$, is there a homomorphism into $\mathbb{B}$ taking $(v_1,\dots,v_n)\notin r^B$ to a nonhyperedge $(a_1,\dots,a_n)\notin r^A$ of $\mathbb{A}$?
}
}

\smallskip
\rule{0cm}{0.4cm}The case of no implied \emph{equalities} is considered in Ham~\cite{ham:SATconf,ham:SAT}, with a complete tractability classification in the case of Boolean constraint languages.

The process of identifying and adding implied constraints is incorporated into typical constraint solvers and can be thought of as gradually collapsing and enriching the given instance toward an instance without implied constraints.  In model-theoretic terms, the structure obtained by adjoining all implied constraints is the \emph{reflection} (or \emph{replica}) of an instance $\mathbb{B}$ into the quasivariety of $\mathbb{A}$; see Maltsev \cite[V.11.3]{mal} and  Sections~\ref{sec:reflection1} and \ref{sec:reflection} below.

\section{Primitive positive formul{\ae} and robust satisfiability}\label{sec:pp}
\begin{defn}\label{defn:pp}
An \emph{atomic formula} is an expression of the form $x= y$ or $(x_1,\dots,x_n)\in r$ for some $r\in \mathscr{R}$.  A \emph{primitive positive formula} \up(abbreviated to \emph{pp-formula}\up) is a formula obtained from a conjunction of atomic formul{\ae} by existentially quantifying some variables.   A pp-formula $\phi(x_1,\dots,x_n)$ with free variables $x_1,\dots,x_n$ defines an $n$-ary relation $r_\phi$, which on an $\mathscr{R}$-structure $\mathbb{A}$ is interpreted as the solution set of $\phi$.  If $\mathscr{F}$ is a set of pp-formul{\ae}, then $\mathbb{A}_\mathscr{F}$ denotes  $\langle A;\{r_\phi\mid \phi\in\mathscr{F}\}\rangle$.
\end{defn}
We let $\pp(\mathscr{R})$ be the set of all pp-formul{\ae} (over some fixed countably infinite set of variables) in $\mathscr{R}$ and let $\pp(\mathbb{A})$ denote the set of all relations on~$A$ that are pp-definable from the fundamental relations of $\mathbb{A}$.  

Let $\mathbb{A},\mathbb{B}$ be $\mathscr{R}$-structures and $\mathscr{F}\subseteq\pp(\mathscr{R})$.  
For a subset $S\subseteq B$, a function $\nu:S\to A$ is \emph{$\mathscr{F}$-compatible} if it is a homomorphism from the substructure $\mathbb{S}$ of $\mathbb{B}_\mathscr{F}$ to $\mathbb{A}_{\mathscr{F}}$.
A function $\nu:S\to A$ extends to a homomorphism precisely when it preserves \emph{all} of $\pp(\mathscr{R})$, so restricting $\mathscr{F}$ to a fixed finite subset of $\pp(\mathscr{R})$ is the natural local condition for extendability; see \cite[Lemma 3.1]{jac:SAT}.
\begin{defn}\label{defn:kf}\cite{ham:SAT} 
Let  $\mathscr{F}$ be a finite set of pp-formul{\ae} in~$\mathscr{R}$ and let $\mathbb{A}$ be a fixed finite $\mathscr{R}$-structure.  For a finite $\mathscr{R}$-structure $\mathbb{B}$, we say that $\mathbb{B}$ is \emph{$(k,\mathscr{F})$-robustly satisfiable} (with respect to $\mathbb{A}$) if $\mathbb{B}$ is a YES instance of $\CSP(\mathbb{A})$ and for every $k$-element subset $S$ of $B$ and every $\mathscr{F}$-compatible assignment $\nu:S\to A$, there is a full solution extending $\nu$.  The structure~$\mathbb{B}$ is $({}{\leq k},\mathscr{F})$-robustly satisfiable if it is $(\ell,\mathscr{F})$-robustly satisfiable for every $\ell\leq k$.
\end{defn} 
Note that $(0,\mathscr{F})$-robust satisfiability coincides with satisfiability.
In \cite{beacul}, the case of $(k,\varnothing)$-robust satisfiability is considered for SAT-related probems using the notation $\widehat{{\bf U}^k}$; this appears in the context of phase transitions and implicit constraints.  The concept of $(k,\varnothing)$-robust satisfiability is called \emph{$k$-supersymmetric} in Gottlob \cite{got}, where it is used to show that there is no polynomial time solver for a minimal constraint network, an issue lingering since the pioneering work of Montanari \cite{mon}.  \ If $\mathscr{P}$ denotes the conjunction-free pp-formul{\ae}, then $(k,\mathscr{P})$-robust satisfiability is the ``$k$-robust satisfiability'' concept introduced in Abramsky, Gottlob and Kolaitis \cite{AGK}, where (for $k=3$ in $3\SAT$) it is applied to show the intractability of detecting local hidden-variable models in quantum mechanics.  The second author showed the \texttt{NP}-completeness of a promise problem form of $(2,\mathscr{P})$-robust satisfiability for positive 1-in-$3\SAT$, and used it to solve a  20+ year old problem in semigroup theory, itself motivated by issues in formal languages.  The first author \cite{ham:SATconf,ham:SAT} recently classified the tractability of $(2,\mathscr{F})$-robust satisfiability (for some $\mathscr{F}$) in the case of Boolean constraint languages.

More complicated sets $\mathscr{F}$ are necessary once $k$ is larger than the arity of relations in~$\mathscr{R}$.  

%
\section{Primitive positive definability and polymorphisms}\label{sec:poly}
When  $\mathscr{R}$ is pp-definable from a set of relations $\mathscr{S}$ on a set $A$ then there is logspace reduction from $\CSP(\langle A;\mathscr{R}\rangle)$ to $\CSP(\langle A;\mathscr{S}\rangle)$. 
This fundamental idea was primarily developed through the work of Cohen, Jeavons and others \cite{Jeav98,JCC,JCG,JCG97,JCP}, though aspects appear in proof of Schaefer's original dichotomy for Boolean CSPs~\cite{sch}.  We investigate this further in Section~\ref{sec:ppstable}. 

There is a well-known Galois correspondence  between sets of relations on a set $A$ and the sets of operations on $A$; see \cite{gei}.   The link is via \emph{polymorphisms}, which are homomorphisms from the direct product $\mathbb{A}^n$ to $\mathbb{A}$.  
In other words, for each relation $r\in\mathscr{R}$ (with arity $k$, say), if we are given an $k\times n$ matrix of entries from $A$, with each column being a $k$-tuple in $r$, then applying the polymorphism~$f$ to each row produces a $k$-tuple of outputs that also must lie in~$r$.
We let $\Pol(\mathbb{A})$ denote the family of all polymorphisms of the relational structure~$\mathbb{A}$.  In general we have $\Pol(\langle A;\mathscr{R}\rangle)\subseteq \Pol(\langle A;\mathscr{S}\rangle)$ if and only if $\mathscr{S}\subseteq \pp(\langle A,\mathscr{R}\rangle)$, so that pp-definability is somehow captured by polymorphisms.  Properties that restrict the strength of pp-definability usually end up being expressible by way of polymorphism equations.  The article \cite{JKN} includes a survey of most of the relevant conditions, but we list just the minimum necessary for our results and arguments.
\begin{itemize}
\item 
An $n$-ary operation $w:A^n\to A$ on a set $A$ is a \emph{weak near unanimity} operation \up(or \emph{WNU}\up) if it satisfies $w(x,x,\dots,x)=x$ (idempotence) and
\(
w(y,x,\dots,x)= w(x,y,\dots,x)=\dots=w(x,x,\dots,y)
\)
for all $x,y$.  A weak near unanimity operation is \emph{near unanimity} \up(NU\up) if it additionally satisfies $w(y,x,\dots,x)=x$.  
\item If the condition of being idempotent is dropped, we refer to a \emph{quasi WNU}, and a \emph{quasi NU} respectively.
\end{itemize}
We mention that most algebraic approaches use the assumption that the template $\mathbb{A}$ is a \emph{core}, meaning that it has no proper retracts.  We now list a selection of pertinent results and conjectures that are expressed in the language of polymorphisms.

The fundamental conjecture on fixed template CSP complexity is the following refinement of Feder and Vardi's original.
\begin{ADC} \cite{BJK}\label{ADC}
Let $\mathbb{A}$ be a finite core relational structure of finite relational signature.  If $\mathbb{A}$ has a  WNU polymorphism then $\CSP(\mathbb{A})$ is tractable.  If $\mathbb{A}$ has no WNU polymorphism then $\CSP(\mathbb{A})$ is \texttt{NP}-complete.
\end{ADC}
The final sentence in the conjecture is proved already in \cite{BJK} (with the WNU condition we state established in \cite{marmck}), with completeness with respect to first order reductions established in \cite{lartes}. There are no counterexamples to the conjecture amongst known classifications.

In the following, \emph{bounded width} corresponds to solvability by way of a local consistency check algorithm, while strict width is a restricted case of this, where every family of locally consistent partial solutions extends to a solution.
\begin{theorem}\label{thm:BW}
Let $\mathbb{A}$ be a finite core relational structure of finite relational signature.\vspace{-.2cm}
\begin{enumerate}
\item \up(Feder and Vardi \cite{fedvar}.\up)  $\CSP(\mathbb{A})$ has strict width if and only if $\mathbb{A}$ has an NU polymorphism.
\item \up(Barto and Kozik \cite{barkoz:BW}.\up) $\CSP(\mathbb{A})$ has bounded width if and only if $\mathbb{A}$ has a $3$-ary WNU $w_3$ and a $4$-ary WNU $w_4$ such that $w_3(y,x,x)=w_4(y,x,x,x)$ holds for all $x,y$.
\end{enumerate}
\end{theorem}

\section{Main results}\label{sec:mainresults}
The main result (ANT~\ref{thm:main}) concerns the following promise problem, which simultaneously extends $\CSP(\mathbb{A})$, $\CSP_\infty(\mathbb{A})$, robust-$\CSP(\mathbb{A})$ \cite{AGK},  $\SEP(\mathbb{A})$  \cite{ham:SATconf,ham:SAT} and others.  In the title line,~$k$ is a non-negative integer and $\mathscr{F}$ is a finite set of pp-formul{\ae} in the signature of $\mathbb{A}$.\\
\fbox{\parbox{0.98\textwidth}
{Promise problem: $(Y_{(k,\mathscr{F}), \mathsf{Q}},N_{\CSP})$ for $\mathbb{A}$.
\begin{description}
\item[YES:] $\mathbb{B}$ is $(k,\mathscr{F})$-robustly satisfiable with respect to $\mathbb{A}$ and has no implied constraints.
\item[NO:] $\mathbb{B}$ is a no instance of $\CSP(\mathbb{A})$.
\end{description}
}}\\
\rule{0cm}{.5cm}We will let $Y_{(k,\mathscr{F})}$ denote the YES promise but where ``no implied constraints'' is omitted.
\begin{AoNT}\label{thm:main}
Let $\mathbb{A}$ be a finite core relational structure in finite signature~$\mathscr{R}$.  \vspace{-.1cm}
\begin{enumerate}
\item \up(Everything is easy.\up) If $\CSP(\mathbb{A})$ is tractable then so also is deciding both $\CSP_\infty(\mathbb{A})$ and $(k,\mathscr{F})$-robust satisfiability, for any $k$ and any finite set $\mathscr{F}$ of pp-formul{\ae}.
\item \up(Nothing is easy.\up) If $\mathbb{A}$ has no WNU, then for all $k$ there exists a finite set of pp-formul{\ae}~$\mathscr{F}$ such that $(Y_{(k,\mathscr{F}),\mathsf{Q}},N_{\CSP})$ is \texttt{NP}-complete for $\mathbb{A}$ with respect to first order reductions.
\end{enumerate}
\end{AoNT}

\begin{remark}\label{rem:PCP}
The ANT \ref{thm:main} shows that the ADC is equivalent to the ostensibly far stronger dichotomy statement: either there is a WNU and (1) holds, or there is no WNU and $(Y_{(k,\mathscr{F}), \mathsf{Q}},N_{\CSP})$  is \texttt{NP}-complete.  
\end{remark}
\begin{remark}
In part (1) of the ANT, we do not require that $\mathbb{A}$ be a core, only that $\CSP(\mathbb{A}_{\rm Con})$ be tractable, where $\mathbb{A}_{\rm Con}$ denotes the result of enlarging the signature of $\mathbb{A}$ to include all singleton unary relations.  When $\mathbb{A}$ is a core, it is known that $\CSP(\mathbb{A}_{\rm Con})$ is tractable if and only if $\CSP(\mathbb{A})$ is tractable.
\end{remark}
The following result gives a dichotomy within tractable complexity classes (cf.~Theorem~\ref{thm:BW}(2)).
\begin{thm}\label{thm:gap}
Let $\mathbb{A}$ be a finite core relational structure in finite signature~$\mathscr{R}$. \vspace{-.1cm}
\begin{enumerate}
\item If $\CSP(\mathbb{A})$ has bounded width, then there exists $n$ such that for all $k\geq n$ and for all finite sets of pp-formul{\ae} $\mathscr{F}$, the promise problem $(Y_{(k,\mathscr{F}), \mathsf{Q}},N_{\CSP})$ lies in $\texttt{AC}^0$. (If $\CSP(\mathbb{A})$ has strict width, then the class of $(k,\mathscr{F})$-robustly satisfiable instances is itself first order definable.) 
\item If $\CSP(\mathbb{A})$ does not have bounded width then for some prime number $p$ and for all $k$ there exists an $\mathscr{F}$ such that $(Y_{(k,\mathscr{F}),\mathsf{Q}},N_{\CSP})$ is $\texttt{Mod}_p(\texttt{L})$-hard.
\end{enumerate}
\end{thm}
Recall that the $\texttt{Mod}_p(\texttt{L})$ class contains $\texttt{L}$ and hence properly contains $\texttt{AC}^0$; the precise relationship with \texttt{NL} is unknown.  Thus
Theorem \ref{thm:gap} shows that in contrast to $\CSP(\mathbb{A})$ (see~\cite{ABISV,lartes}), one cannot get \texttt{L}-completeness, nor \texttt{NL}-completeness for $(Y_{(k,\mathscr{F}),\mathsf{Q}},N_{\CSP})$ over $\mathbb{A}$ unless there are unexpected collapses between $\texttt{L}$, $\texttt{NL}$ and  $\texttt{Mod}_p(\texttt{L})$ for various $p$.

The complexity of $\CSP(\mathbb{A})$ is determined by the core retract of $\mathbb{A}$, but this is not true for $(k,\mathscr{F})$-robust satisfiability and quasivariety membership; see \cite{jac:SAT} and \cite{ham:SATconf,ham:SAT} for example.  The remaining results however apply regardless of whether $\mathbb{A}$ itself is a core.

\begin{cor}\label{cor:qvar}
Let $\mathbb{A}$ be finite relational structure of finite signature.
\begin{itemize}
\item If $\mathbb{A}$ has no quasi WNU polymorphism then $\CSP_{\infty}(\mathbb{A})$ is \texttt{NP}-complete,
\item If the core retract of $\mathbb{A}$ fails to have bounded width then $\mathsf{Q}(\mathbb{A})$ is not finitely axiomatisable in first order logic, even amongst finite structures.
\end{itemize}
\end{cor} 
The second statement is equivalent to the absence of quasi WNUs satisfying the conditions of Theorem \ref{thm:BW}(2).
Similar statements to Corollary \ref{cor:qvar} hold for problems intermediate to $\CSP(\mathbb{A})$ and $\CSP_{\infty}(\mathbb{A})$, such as the $\SEP(\mathbb{A})$ of \cite{ham:SAT} and the problem of detecting if no variable is nontrivially forced to take a fixed value: variables with implicitly fixed values have been called the ``backbone'' or ``frozen variables''; see \cite{jonkro} for example.

The following corollary simultaneously covers the original Hell-Ne\v{s}et\v{r}il Dichotomy for simple graphs and a corresponding quasivariety dichotomy; again it does not assume cores.
\begin{GDfSG}\label{thm:GDfSG}
Let $\mathbb{G}$ be a finite simple graph. 
\begin{enumerate}
\item If $\mathbb{G}$ is bipartite, then deciding $\CSP(\mathbb{G})$ and deciding membership in the quasivariety of $\mathbb{G}$ are both tractable.  
\item Otherwise, the following promise problem is \texttt{NP}-complete with respect to first order reductions and for finite input graph $\mathbb{H}$\up: 
\begin{description}
\item[Yes] $\mathbb{H}$ is in $\mathsf{Q}^+(\mathbb{G})$.
\item[No] $\mathbb{H}$ has no homomorphisms into $\mathbb{G}$.
\end{description}
\end{enumerate}
\end{GDfSG}

The first author \cite{ham:SATconf,ham:SAT} classified the tractability of $(2,\mathscr{F})$-robust satisfiability and other problems for Boolean constraint languages, in the form of a ``Gap Trichotomy Theorem''.   By employing the ANT \ref{thm:main} to the same family of critical problems to those examined in \cite{ham:SATconf,ham:SAT}, we can show that the same trichotomy holds with $2$ replaced by any $k\geq 0$ and ``separation condition'' replaced by quasivariety membership.  Using Theorem \ref{thm:gap} it is further possible to classify the precise fine level complexity of  $(Y_{(k,\mathscr{F}),\mathsf{Q}},N_{\CSP})$ for sufficiently large $k$ and $\mathscr{F}$: it is either trivial, or nontrivial but in $\texttt{AC}^0$, or is $\oplus L$-complete or \texttt{NP}-complete.  The details will be given in subsequent work.  Further results that will be given include a complete classification of when the relational clone of a finite semigroup has tractable quasivariety membership problem.

To complete this section we give an overview of how the proof of the ANT develops across the remaining sections.  Part (1) of the ANT is quite straightforward and follows a simpler but similar argument in \cite{jac:SAT}. We sketch the details in Section \ref{sec:proof}.  The proof of the ANT part (2) involves taking the usual proof that $\CSP(\mathbb{A})$ is \texttt{NP}-complete when $\mathbb{A}$ has no WNU, but at each step explaining how the infinitely stronger promise $(Y_{(k,\mathscr{F}),\mathsf{Q}},N_{\CSP{}})$ can be carried through for some suitably constructed $\mathscr{F}$.  There are five main steps which are developed as separate sections once we have introduced some further preliminary development.  The various stages of the proof are unified in Section \ref{sec:proof}.  An outline of the proof of Theorem \ref{thm:gap} is given in Section \ref{sec:BW}, while Section \ref{sec:discuss} gives some ideas for future work, including an example demonstrating the limits to which the NO promise provided by the PCP Theorem can be incorporated in the ANT.

\section{Preliminary development: $\mathscr{F}$-types and claw formul{\ae}}\label{sec:Ftype}
We now establish some useful preliminary constructions relating to pp-formul{\ae} and $(k,\mathscr{F})$-robustness.  Throughout, $\mathbb{A}$ and $\mathscr{F}\subseteq\pp(\mathscr{R})$ are fixed and $\mathbb{B}$ is an input $\mathscr{R}$-structure; all are~finite.

Let $x_1,\dots,x_n$ denote the free variables in some pp-formula $\phi(x_1,\dots,x_n)\in \mathscr{F}$ and let $k$ be a nonnegative integer.  For any function $\iota:\{x_1,\dots,x_n\}\to\{x_1,\dots,x_k\}$ we let $\phi^\iota(x_1,\dots,x_k)$ denote the formula $\phi(\iota(x_1),\dots,\iota(x_n))$.  We let $\mathscr{F}_k$ denote the set of all formul{\ae} obtained in this way.  The following is immediate.
\begin{lem}\label{lem:Fk}
Let $\mathbb{A}$ and $\mathbb{B}$ be $\mathscr{R}$-structures and consider a subset~$\{b_1,\dots,b_k\}$ of~$B$.  A function $\nu\colon\{b_1,\dots,b_k\}\to A$ is $\mathscr{F}$-compatible if and only if for every  $\phi(x_1,\dots,x_k)\in\mathscr{F}_k$, if $\mathbb{B}\models\phi(b_1,\dots,b_k)$ then $\mathbb{A}\models \phi(\nu(b_1),\dots,\nu(b_k))$.
\end{lem}
The following gives a natural restriction of the model theoretic ``$k$-type'' to pp-formul{\ae}.
\begin{defn}
Let $\mathscr{R}$ be a finite relational signature and $\mathscr{F}$ a set of pp-formul{\ae} in $\mathscr{R}$.  \begin{enumerate}
\item A \emph{$(k,\mathscr{F})$-type} is any conjunction of distinct $k$-ary formul{\ae} in $\mathscr{F}_k$. The set of all $(k,\mathscr{F})$-types is denoted by $\type_k(\mathscr{F})$.  
\item The \emph{$(k,\mathscr{F})$-type of a tuple $\vec{b}\in B^k$} is the conjunction 
\(\displaystyle
\Bigand_{
\footnotesize\begin{matrix}
\phi(x_1,\dots,x_k)\in \mathscr{F}_k\\
\mathbb{B}\models \phi(\vec{b})
\end{matrix}
}
\phi(\vec{x})
\)
and is denoted $\tau_{\vec{b}}^\mathscr{F}(x_1,\dots,x_k)$ (with $\mathbb{B}$ implicit).
\item For $\ell\leq k$ we let $\mathscr{F}{\mid}_\ell$ denote
\(
\{\exists x_{\ell+1}\dots \exists x_{k}\ \tau(x_1,\dots,x_{k})\mid \tau\in \type_k(\mathscr{F})\}.
\)
\end{enumerate}
\end{defn}
The following follows immediately from Lemma \ref{lem:Fk} and the definition of $(k,\mathscr{F})$-types.
\begin{lem}\label{lem:ktype}
Let $\mathbb{A}$ and $\mathbb{B}$ be $\mathscr{R}$-structures and consider a $k$-element subset $\{b_1,\dots,b_k\}$ of $B$.
A partial map $\nu\colon\{b_1,\dots,b_k\}\to A$ is $\mathscr{F}$-compatible if and only if  $\mathbb{A}\models \tau^\mathscr{F}_{\vec{b}}(\nu(b_1),\dots,\nu(b_k))$.
\end{lem}
The next lemma has a straightforward proof, which is given in the appendix.
\begin{lem}\label{lem:leqk}
Let $\mathbb{A}$ and $\mathbb{B}$ be finite $\mathscr{R}$-structures and let $\mathscr{F}$ be a finite set of pp-formul{\ae} in terms of $\mathscr{R}$.  If $\mathbb{B}$ is $(k,\mathscr{F})$-robustly satisfiable into $\mathbb{A}$ and $\ell\leq k$, then $\mathbb{B}$ is $(\ell,\mathscr{F}{\mid}_\ell)$-robustly satisfiable.  In particular, if $\mathbb{B}$ is $(k,\mathscr{F})$-robustly satisfiable for some finite set of pp-formul{\ae}~$\mathscr{F}$, then $\mathbb{B}$ is $({\leq k},\bigcup_{0\leq i\leq k}\mathscr{F}{\mid}_i)$-robustly satisfiable, where $\bigcup_{0\leq i\leq k}\mathscr{F}{\mid}_i$ is also a finite set of pp-formul{\ae}.
\end{lem}

Recall from Section \ref{sec:poly} that when  $\mathscr{R}$ is pp-definable from a set of relations $\mathscr{S}$ on a set $A$ then there is logspace reduction from $\CSP(\langle A;\mathscr{R}^A\rangle)$ to $\CSP(\langle A;\mathscr{S}^A\rangle)$.  Assume then that each relation symbol $r\in \mathscr{R}$ has been matched to some fixed defining $\mathscr{S}$-formula $\rho_r(x_1,\dots,x_n)$ of the same arity $n$ as $r$:
\[
\exists y_1\dots \exists y_m\ \Bigand_{1\leq i\leq k}\alpha_i(x_{i,1},\dots,x_{i,n_i},y_{i,1},\dots,
y_{i,m_i}),\tag{$\dagger$}\label{eqn:pp}
\]
where each $\alpha_i$ is an atomic formula in $\mathscr{S}\cup\{=\}$, and  
$\bigcup_{1\leq i\leq k}\{x_{i,1},\dots,x_{i,n_i}\}=\{x_1,\dots,x_n\}$ and $\bigcup_{1\leq i\leq k}\{y_{i,1},\dots,y_{i,m_i}\}=\{y_1,\dots,y_m\}$.  Let $\rho_r^\flat$ denote the underlying open formula obtained from $\rho_r$ by removing quantifiers: variables of $\rho_r^\flat$ that are quantified in $\rho_r$ will be called \emph{existential variables} (or $\exists$-variables: the $y_i$ in \eqref{eqn:pp}) and the other variables will be referred to as \emph{open variables}.

Each pp-formula $\psi(x_1,\dots,x_\ell)$ in the signature $\mathscr{R}$  becomes a pp-formula $\psi^\mathscr{S}(x_1,\dots,x_\ell)$ in the signature $\mathscr{S}$: replace each conjunct in $\psi$---an atomic formula $r(x_1,\dots,x_n)$ in for some $r\in\mathscr{R}$---by the defining formula $\rho_r$ as in \eqref{eqn:pp}, and then apply the usual logical rules for moving quantifiers to the front (including renaming quantified variables where necessary).  

\begin{definition}\label{defn:claw}
Let $\mathscr{S}$ define $\mathscr{R}$ by pp-formul{\ae} $\{\rho_r\mid r\in \mathscr{R}\}\subseteq \pp(\mathscr{S})$.
Let $k,\ell$ be fixed non-negative integers and $\mathscr{F}$ a finite set of pp-formul{\ae} in $\mathscr{R}$.  A \emph{claw formula} for $\mathscr{F}$ of arity $k$ and bound $\ell$ is any pp-formula in $\mathscr{S}$ of the form constructed in the third step below:
\begin{enumerate}
\item (The \emph{talon}.) Let $\gamma$ denote any conjunction $\bigand_{1\leq i\leq k'}\rho_{r_i}^\flat$, where $r_i\in\mathscr{R}$ and $k'\leq k$.  We allow some identification between open variables, but not between existential variables.
\item (The \emph{wrist}.) Let $\sigma$ be an $(\ell',\mathscr{F})$-type in $\mathscr{R}$ for some $\ell'\leq \ell$.  Some of the $\ell'$ free variables in $\sigma$ may be identified with open variables in $\gamma$, but not with existential variables.
\item (The \emph{claw}.) Existentially quantify all but $k$ of the unquantified variables in the conjunction $\gamma\And \sigma^\mathscr{S}$.
\end{enumerate}
\end{definition}

\section{Step 1.  Reflection}\label{sec:reflection1}
\begin{defn}
Let $\mathbb{A}$, $k$ and $\mathscr{F}$ be fixed.  For an input $\mathscr{R}$-structure $\mathbb{B}$, let $\mathbb{B}^\downarrow$ be the result of adjoining all hyperedges to $\mathbb{B}$ that are implied by $\mathscr{F}$-compatible assignments from subsets of $\mathbb{B}$ on at most $k$ elements.  The structure $\mathbb{B}^\downarrow$ will be called the \emph{1-step $(k,\mathscr{F})$-reflection} of $\mathbb{B}$.
\end{defn}

Note that implied constraints can involve equality: so $\mathbb{B}^\downarrow$ is in general a proper quotient of $\mathbb{B}$.  Under the promise $(Y_{k,\mathscr{F}},N_{\CSP{}})$ it is possible to show that there is a first order query that defines~$\mathbb{B}^\downarrow$; again, this is given in Section \ref{sec:reflection}, where some further development of reflections is given.   To prove the basic results with respect to polynomial time (or even logspace reductions) we need only the following  lemma and the observation that $\mathbb{B}^\downarrow$ can be constructed in polynomial time from $\mathbb{B}$, because there are only polynomially many $\mathscr{F}$-compatible assignments from subsets of size at most $k$. 
\begin{lem}\label{lem:reflect}
If $\mathbb{B}$ is $({\leq} k,\mathscr{F})$-robustly satisfiable with respect to $\mathbb{A}$, and $k\geq\arity(\mathscr{R})$, then $\mathbb{B}^\downarrow$ lies in the quasivariety of $\mathbb{A}$ and is also $({\leq} k,\mathscr{F})$-robustly satisfiable.  If $\mathbb{B}$ is a NO instance of $\CSP(\mathbb{A})$ then so also is $\mathbb{B}^\downarrow$.
\end{lem}

\section{Step 2.  Stability of robustness over primitive-positive reductions}\label{sec:ppstable}
We prove the following promise problem variant of the usual pp-reduction for CSPs.
\begin{thm}\label{thm:ppGAP}
Assume that $\mathbb{A}_1=\langle A,\mathscr{R}^A\rangle$ and $\mathbb{A}_2=\langle A,\mathscr{S}^A\rangle$ are two relational structures on the same finite set $A$, with $\mathscr{R}^A\subseteq \pp(\mathbb{A}_2)$ finite and $\ell:=\arity(\mathscr{R})$.  Let $\mathscr{F}$ be a finite set of pp-formul{\ae} in the language of $\mathscr{R}$.  Then, for any $k$,  the standard pp-reduction of  $\CSP(\mathbb{A}_1)$ to $\CSP(\mathbb{A}_2)$ takes $({\leq} k \ell,\mathscr{F})$-robustly satisfiable instances of $\CSP(\mathbb{A}_1)$ to $(k,\mathscr{G})$-robustly satisfiable instances of $\CSP(\mathbb{A}_2)$, where $\mathscr{G}$ denotes the $k$-ary claw formul{\ae} for $\mathscr{F}$ of bound~$k \ell$.
\end{thm}
First briefly recall the precise nature of the ``standard reduction'' described in Theorem~\ref{thm:ppGAP}.  Recall that each $r\in \mathscr{R}$ corresponds to an $\mathscr{S}$ formula $\rho_r$, as in~\eqref{eqn:pp}.
For an instance $\mathbb{B}=\langle B;\mathscr{R}^B\rangle$ of $\CSP(\mathbb{A}_1)$, an instance $\mathbb{B}^\sharp$ of $\CSP(\mathbb{A}_2)$ is constructed in the following way.  For each hyperedge $(b_1,\dots,b_n)\in r$ in $\mathbb{B}$ (and adopting the generic notation of \eqref{eqn:pp}), new elements $c_1,\dots,c_m$ are added to the universe of $B$, and the hyperedge $(b_1,\dots,b_n)\in r$ is replaced by the hyperedges $\alpha_i(b_{i,1},\dots,b_{i,n_k},c_{i,1},\dots,c_{i,m_i})$ for each $i=1,\dots,k$.  
Note that new elements $c_1,\dots,c_m$ are introduced for every instance of a hyperedge.  The new elements will be referred to as \emph{existential elements} (or $\exists$-elements), and for any $D\subseteq B^\sharp$ we let $D_\exists$ denote $\exists$-elements in $D$.   Elements of $B$ will be referred to as \emph{open elements}, and we write $D_B$ for $D\cap B=D\backslash D_\exists$. 

It is easy to see that there is a homomorphism from $\mathbb{B}$ to $\mathbb{A}_1$ if and only if there is one from $\mathbb{B}^\sharp$ to $\mathbb{A}_2$: this is the usual logspace CSP reduction, which is a first order reduction when none of the $\rho_r$ formul{\ae} involve equality \cite{lartes}.  Now assume that $\mathbb{B}$ is $(\leq k\ell,\mathscr{F})$-robustly satisfiable with respect to $\mathbb{A}_1$ and consider a $k$-element subset $D\subseteq B^\sharp$, for which there is a $\mathscr{G}$-compatible assignment into $A$.  The following arguments will refer back to the 3-step construction of claw formul{\ae} in Definition \ref{defn:claw}.

Each $c\in D_\exists$ was introduced in replacing a hyperedge of $\mathbb{B}$ in signature $\mathscr{R}$ by a \emph{family} of hyperedges in the signature $\mathscr{S}$, according to the pp-definition as in \eqref{eqn:pp}.  Each element of $D_\exists$ appears in at most one such family of $\mathscr{S}$-hyperedges, so the number of these, $k'$, is at most $|D_\exists|\leq k$.  
Observe that these hyperedge families correspond to an interpretation of a conjunction $\gamma$ of $k'$ many formul{\ae} as in step 1 of Definition \ref{defn:claw}: there is no identification of $\exists$-elements, but there may be of open elements.  
Each of these families involves at most~$\ell$ open elements, so that at most $k'\times \ell$ open elements appear in these hyperedge families.  
Let $O_B$ denote these elements.  Because $k'+|D_B|\leq |D_\exists|+ |D_B|=k$ and $|O_B|\leq k'\ell$, we have $|O_B\cup D_B|\leq k'\ell+|D_B|\leq k\ell$.  Let $\sigma$ denote the $(\ell',\mathscr{F})$-type of $O_B\cup D_B$ in $\mathbb{B}$, as in the second step of Definition \ref{defn:claw}.  (Here we treat $O_B\cup D_B$ as a tuple ordered in any fixed way.)  Observe that some elements $b$ of $D_B$ may also lie in $O_B$, and  we will assume then that the variable in $\sigma$ corresponding to $b$ has been identified with the variable in $\gamma$  corresponding to $b$.  Let $U$ be the set of all unquantified variables in $\gamma\And \sigma^\mathscr{S}$ that do not correspond to elements of $D$.  
The claw formula $\exists U\ \gamma\And \sigma^\mathscr{S}$ is in $\mathscr{G}$ and is satisfied by $\mathbb{B}^\sharp$ at $D$ (again, arbitrarily treated as a tuple).  
Hence $\exists U\ \gamma\And \sigma^\mathscr{S}$ is preserved by $\nu$.  In particular then, in $\mathbb{A}_2$ we can find values for the variables corresponding to the elements of $O_B$ that witness the satisfaction of  $\exists U\ \gamma\And \sigma^\mathscr{S}$ at $\nu(D)$.  
Let $\nu':O_B\cup D\to A$ be the extension of $\nu$ obtained by giving elements of $O_B\backslash D_B$ these witnessing values.  
Because $\sigma$ is the $(\ell',\mathscr{F})$-type of $O_B\cup D_B$, it follows from Lemma~\ref{lem:ktype} that $\nu'{\mid}_{O_B\cup D_B}$ is $\mathscr{F}$-compatible, so by the assumed $({\leq} k \ell,\mathscr{F})$-robust satisfiability of $\mathbb{B}$ it follows that $\nu'{\mid}_{O_B\cup D_B}$ extends to a homomorphism $\nu^+$ from $\mathbb{B}$ to $\mathbb{A}_1$.  
By the usual pp-reduction, $\nu^+$ extends to a homomorphism $\nu^\sharp$ from $\mathbb{B}^\sharp$ to $\mathbb{A}_2$.  
Now $\nu^\sharp$ agrees with $\nu$ on $D_B$, but also, we may assume that it agrees with $\nu$ on $D_\exists$, because the values given $O_B$ by $\nu'$ (and hence $\nu^\sharp$) were such that $\gamma$ held.
Thus we have extended $\nu$ to a homomorphism, as required.  (An example of this argument is given in Appendix \ref{eg:ppgap}.)

\section{Step 3. $(k,\varnothing)$-robustness of $(3k+3)\SAT$}\label{sec:3k3SAT}
Gottlob \cite[Lemma 1]{got} showed that the standard Yes/No decision problem $3\SAT$ reduces to the promise $(Y_{(k,\varnothing)},N_{\CSP{}})$ for $(3k+3)\SAT$.  For the sake of completeness of our sketch, we recall the basic idea.  The construction is to replace in a $3\SAT$ instance $\mathbb{B}$, each element~$b$ by $2k+1$ copies $b_{1},\dots,b_{2k+1}$ and then each clause $(b\vee c\vee d)$ by all $\binom{2k+1}{k+1}^3$ clauses of the form $({b}_{i_1}\vee\dots \vee {b}_{i_{k+1}}\vee {c}_{i_1'}\vee\dots \vee {c}_{i_{k+1}'}\vee {d}_{i_1''}\vee\dots \vee {d}_{i_{k+1}''})$ where the $i_j, i_j',i_j''$ are from $\{1,\dots,2k+1\}$.  No assignment on $k$ elements covers all of the $k+1$ copies of any element in a clause it appears, which enables the flexibility for such assignments to always extend to a solution, provided (and only when) $\mathbb{B}$ is a YES instance.  It is routine to achieve this via a first order reduction: the argument is given in Section \ref{appsec:3k3SAT}.

\section{Step 4. $(k,\mathscr{F})$-robustness of $3\SAT$}\label{sec:3SAT}
We now establish the following theorem by reduction from the result in Step 3.  Critically, the value of $k$ is arbitrary, but the constraint language ($3\SAT$) has fixed arity $3$.
\begin{thm}\label{thm:3SAT}
Fix any $k\geq 0$ and let $\mathscr{F}$ be the set of all claw formul{\ae} for $\varnothing$ of arity $k$ and with bound $k$.  Then  $(Y_{(k,\mathscr{F})},N_{\CSP})$ for $3\SAT$ is \texttt{NP}-complete via first order reductions.
\end{thm}
The usual reduction of $n\SAT$ to $3\SAT$ (as in \cite{garjoh} for example) is an example of a pp-reduction, because the $n\SAT$ clause relation $(x_1\vee\dots\vee x_n)$  (where the $x_i$ can be negated variables if need be) is equivalent to the following pp-formula over $n-2$ clause relations of  $3\SAT$:
\[
\exists y_1\dots \exists y_{n-4}\ (x_1\vee x_2\vee y_1)\And\left(\bigand_{3\leq i\leq n-2} (\neg y_{i-2}\vee x_i\vee y_{i-1})\right)\And (\neg y_{n-3}\vee x_{n-1}\vee x_n)\tag{$\ddagger$}\label{eqn:3SAT}
\]
As we are dealing with the standard pp-reduction, an instance $\mathbb{B}$ of $(3k+3)\SAT$ is satisfiable if and only if the constructed instance $\mathbb{B}^\sharp$ of $3\SAT$  is satisfiable.

Now assume that $\mathbb{B}$ is a $(k,\varnothing)$-robustly satisfiable instance of $(3k+3)\SAT$.  Assume  $D$ is a $k$-set from $\mathbb{B}^\sharp$ and $\nu:D\to\{0,1\}$ an $\mathscr{F}$-compatible partial assignment.  As in the proof of Theorem \ref{thm:ppGAP}, there are $k'\leq |D_\exists|$ different clause families involving elements from~$D_\exists$; let $F$ denote this set of families of clauses (each family arising by the replacement of a $(3k+3)\SAT$ clause by the $3k+1$ distinct $3\SAT$ clauses).   Let $\gamma$ denote the conjunction of~$k'$ many pp-formul{\ae} corresponding to these $F$: it is a conjunction of $k'$ distinct copies of the underlying open formula of~\eqref{eqn:3SAT}, possibly with some of the open variables in different copies identified.  Let $U$ be the variables of $\gamma$ that do not correspond to an element of $D$.  Then $\exists U\ \gamma$ is a claw formul{\ae} in the sense of Definition \ref{defn:claw} because the only $(\ell,\varnothing)$-types (as detailed in step 2 of Definition \ref{defn:claw}) are empty formul{\ae}.  This formula $\exists U\ \gamma$ is obviously satisfied at~$D$ in~$\mathbb{B}^\sharp$, so is preserved by $\nu$.  Now the proof deviates from Theorem \ref{thm:ppGAP}.  We show how to assign values to at most $k$ of the remaining open elements of $F$ such that any extension to a full solution on $\mathbb{B}$ extends to one for $\mathbb{B}^\sharp$ in a way consistent with the values given to $D_\exists$ by~$\nu$.

We introduce an arrow notation to help identify how to select the new open elements. 
\begin{itemize}
\item Above the leftmost bracket of the clause family we place a right arrow $\mapsto$, and dually a $\mapsfrom$ over the rightmost bracket.
\item Place a left arrow $\mapsfrom$ above a consecutive pair of brackets ``$)($'' if the  $\exists$-element immediately preceding it is given $0$ by~$\nu$, and dually, $\mapsto$ if the $\exists$-element is assigned $1$.
\end{itemize}
Let us say that two such arrows are \emph{convergent} if they point toward one another.  To extend~$\nu$ to a full solution, it suffices to find, within each pair of convergent arrows, a open-literal to assign the value $1$.  We first give  an example, consisting of a clause family, an assignment to some elements (say, $D_\exists=\{b_1,b_2,b_3\}$ and $D_B=\{a_1\}$) and the arrows placed as determined by the rules:
\[
\null\hspace{1cm}\begin{tabular*}{\columnwidth}{@{\extracolsep{0pt}}ccccccccccccccccc}
${(}$&$a_1$&$a_2$&$b_1$&$)($&$\neg b_1$&$a_3$&$b_2$&$)($& $\neg b_2$&$a_4$&$b_3$&$)($& $\neg b_3$&$a_5$&$a_6$&$)$\\
${(}$&$0$&$a_2$&$0$&$)($&$1$&$a_3$&$1$&$)($& $0$&$a_4$&$0$&$)($& $1$&$a_5$&$a_6$&$)$\\
\rule{0cm}{.55cm}$\stackrel{\mapsto}{(}$&&$a_2$&&$\stackrel{\mapsfrom}{)(}$&&$a_3$&& $\stackrel{\mapsto}{)(}$&&$a_4$&&$\stackrel{\mapsfrom}{)(}$&&$a_5$&$a_6$&$\stackrel{\mapsfrom}{)}$\\
\end{tabular*}
\]
By calling on witnesses to preservation of $\mathscr{F}$ by $\nu$ we can select open literals and values (here $\nu(a_4)=1$ and $\nu(a_2)=1$) that are consistent with the values assigned to~$D_\exists$.  

In the general case: because $\nu$ preserves the claw formula $\exists U\ \gamma$, the $2$-element template for $3\SAT$ has witnesses to all quantified variables.  For each pair of convergent arrows under the assignment by $\nu$ for $D_\exists$, there is a witness to one of the open variables in $\gamma$ taking the value $1$; only one such witness is required for each pair of convergent arrows.  Let~$E$ consist of the  open elements in $F$ corresponding to the selected  witnesses, and extend $\nu$ to~$E$ by giving them the witness values.  Note that  $|E|\leq |D_\exists|$, so that $|E\cup D_B|\leq k$.  Thus $\nu{\mid}_{E\cup D_B}$ extends to a solution for $\mathbb{B}$.  This solution extends to a solution for $\mathbb{B}^\sharp$ in a way that is consistent with the values given elements of $D_\exists$ by $\nu$.

%


\section{Step 5: Idempotence and the algebraic method}
A key development in the algebraic method for CSP complexity was  restriction to idempotent polymorphisms \cite{BJK}.  We now sketch how this works for the $(Y_{(k,\mathscr{F})},N_{\CSP{}})$ promise.

Let $\mathscr{R}_{\Con}$ be the signature obtained by adding a unary relation symbol $\underline{a}$ for each element $a$ of $A$, and let $\mathbb{A}_{\Con}$ denote the structure $\langle A;\mathscr{R}_{\Con}\rangle$, with $\underline{a}$ interpreted as $\{a\}$.
\begin{thm}\label{thm:unaryred}
Let $\mathbb{A}$ be a core and $\mathscr{F}$ be a finite set of pp-formul{\ae} in the language of~$\mathscr{R}_{\Con}$.  Then for any $k$, there exists a finite set $\mathscr{G}$ of pp-formul{\ae} in the language of~$\mathscr{R}$ such that the standard reduction from $\CSP(\mathbb{A}_{\Con})$ to $\CSP(\mathbb{A})$ takes $(k, \mathscr{F})$-robustly satisfiable instances of $\CSP(\mathbb{A}_{\Con})$ to the $(k, \mathscr{G})$-robustly satisfiable instances of $\CSP(\mathbb{A})$. 
\end{thm} 
\begin{proof}[Proof sketch]
Let $\mathbb{B}$ be an instance of $\CSP(\mathbb{A}_{\Con})$.
The standard reduction---known to be first order \cite[Lemma~2.5]{lartes}---involves adjoining a copy of $\mathbb{A}$ to  the instance $\mathbb{B}$, and replacing all hyperedges $b\in \underline{a}$ by identifying $b$ with the adjoined copy of $a$; call this $\mathbb{B}^\sharp$.  (A $(k,\mathscr{F})$-reflection, via the first order version of Lemma \ref{lem:reflect}, can be used to circumvent some technical issues regarding identification of elements.)  Our task is to show how to construct $\mathscr{G}$.  Let $\diag({\mathbb A})$ denote the equality-free positive atomic diagram of ${\mathbb{A}}$ on some set of variables $\{v_a\mid a\in A\}$: the conjunction of all hyperedges of $\mathbb{A}$ considered as atomic formul{\ae}.  Every ${k}$-ary pp-formula $\rho$ in $\mathscr{R}_{\Con}$ becomes a pp-formula $\rho^{\mathscr{R}}$ in $\mathscr{R}$ by identifying variables in the following way.  First take the conjunction of $\rho$ with $\diag({\mathbb A})$.  Then, replace each conjunct of the form $x\in \underline{a}$ in $\rho$, by $x=v_a$.  

Assume $\mathbb{B}$ is $(k,\mathscr{F})$-robustly satisfiable with respect to $\mathbb{A}_{\Con}$ and consider a $\mathscr{G}$-compatible assignment $\nu$ from some $k$-set in $\mathbb{B}^\sharp$. Because $\mathbb{A}$ is a core, there is an automorphism $\alpha$ of $\mathbb{A}$ mapping witnesses to $\diag({\mathbb A})$ to their named location  (that is, taking $v_a$ to $a$).  Then $\alpha\circ \nu$ is $\mathscr{F}$-compatible into $\mathbb{A}_{\Con}$, hence extends to a homomorphism $\psi$ from $\mathbb{B}$.  Then $\alpha^{-1}\circ\psi$ is a homomorphism from $\mathbb{B}^\sharp$ to $\mathbb{A}$ extending $\nu$.
\end{proof}
\section{Proof of ANT~\ref{thm:main} and corollaries}\label{sec:proof} 
\begin{proof}[Proof of ANT] For part (1), we extend an idea from \cite{jac:SAT}. Our proof will use only the assumption that  $\CSP(\mathbb{A}_{\Con})$ is tractable.  This is always true if $\mathbb{A}$ is a core with $\CSP(\mathbb{A})$ tractable.  Now observe that an $\mathscr{F}$-compatible partial assignment $\nu: b_i\mapsto a_i$ from a subset $\{b_1,\dots,b_k\}$ of an instance $\mathbb{B}$ into $\mathbb{A}$ extends to a solution if and only if the structure obtained from $\mathbb{B}$ by adjoining the constraints $\{b_i\in\{a_i\}\mid i=1,\dots,k\}$ is a YES instance of $\CSP(\mathbb{A}_{\Con})$.  Thus after polynomially many calls on the tractable problem $\CSP(\mathbb{A}_{\Con})$, we can decide the $(k,\mathscr{F})$-robust satisfiability of~$\mathbb{B}$.  An almost identical argument will determine if $\mathbb{B}$ has no implied constraints, thus deciding $\CSP_\infty(\mathbb{A})$.

Now to prove ANT part (2).  Let ${\bf A}$ denote the polymorphism algebra of~$\mathbb{A}_{\Con}$.  One of the fundamental consequences of the algebraic method is that if~$\mathbb{A}$ has no WNU polymorphism, then the polymorphism algebra of $3\SAT$ is a homomorphic image of a subalgebra of ${\bf A}$ (direct powers are not required; see \cite[Prop 3.1]{val}).  For CSPs, these facts will give a first order reduction from $3\SAT$ to some finite set of relations $\mathscr{S}^A$ in $\pp(\mathbb{A}_{\Con})$: see \cite{lartes}.  The first step of this reduction is to reduce through homomorphic preimages and subalgebras.  Ham~\cite[\S8]{ham:SAT} showed that these initial reductions also preserve the $(Y_{(\ell,\mathscr{F})},N_{\CSP{}})$ promise, with only minor modification to $\mathscr{F}$.  Combining this with Theorem \ref{thm:3SAT} then Lemma \ref{lem:leqk} we find that for all $\ell$ there exists an $\mathscr{F}_2$ such that $(Y_{({\leq}\ell,\mathscr{F}_2)},N_{\CSP{}})$ is \texttt{NP}-complete for $\langle A,\mathscr{S}^A\rangle$.  Then (using $\ell=\arity(\mathscr{S})\times k$) we can use Theorem \ref{thm:ppGAP} then Lemma \ref{lem:leqk} to find that for every $k$ there exists $\mathscr{F}_3$ such that $(Y_{(\leq k,\mathscr{F}_3)},N_{\CSP{}})$ is \texttt{NP}-complete for $\mathbb{A}_{\Con}$ with respect to first order reductions.  By Theorem \ref{thm:unaryred} the same is true for $\mathbb{A}$, with an amended compatibility condition $\mathscr{F}$ depending on $k$.  Lemma \ref{lem:reflect} then extends the promise to $(Y_{(k,\mathscr{F}),\mathsf{Q}},N_{\CSP{}})$, as required.  \end{proof}
\begin{proof}[Proof of Corollary \ref{cor:qvar}]
Let $\mathbb{A}$ be a finite relational structure without a quasi WNU polymorphism.  By Chen and Larose \cite[Lemma~6.4]{chelar} the core retract $\mathbb{A}^\flat$ of $\mathbb{A}$ has no WNU. Hence the ANT~\ref{thm:main} applies to~$\mathbb{A}^\flat$.  Now $\mathsf{Q}^+(\mathbb{A})$ contains $\mathsf{Q}^+(\mathbb{A}^\flat)$, which contains the YES promise in the ANT~\ref{thm:main} and is disjoint from the NO promise.  Hence membership in $\mathsf{Q}^+(\mathbb{A})$ is \texttt{NP}-complete with respect to first order reductions, and hence is also not finitely axiomatisable in first order logic, even at the finite level.  The same argument using Theorem \ref{thm:gap}(2) implies non-finite axiomatisability in the case that $\mathbb{A}^\flat$ does not have bounded width.
\end{proof}
\begin{proof}[Proof of the Gap Dichotomy for Simple Graphs  \ref{thm:GDfSG}]
If $\mathbb{G}$ is bipartite, then $\CSP(\mathbb{G})$ is tractable and so is deciding membership in $\mathsf{Q}^+(\mathbb{G})$: there are only five distinct quasivarieties \cite{cai,nespul}.  Otherwise, $\mathbb{G}$ is not bipartite and so neither is its core retract.  Hence $\mathbb{G}$ has no quasi WNU; see \cite{BKN}.  Then apply Corollary \ref{cor:qvar}.
\end{proof}
\section{Proof of Theorem \ref{thm:gap}}\label{sec:BW}
Due to space constraints we give only a brief overview of the method.  
A CSP has bounded width provided that there exists $j$ such that the existence of a homomorphism from $\mathbb{B}$ to~$\mathbb{A}$ is equivalent to a family of partial homomorphisms on all subsets of size at most $j+1$, with the family satisfying a compatibility condition, known as a \emph{$(j,j+1)$-strategy}; see \cite{barkoz:BW}.  When $k>j$ and input $\mathbb{B}$ satisfies the $Y_{(k,\mathscr{F}), \mathsf{Q}}$ promise, there is an obvious choice for a $(j,j+1)$-strategy: the family of all maps that can extend to $\mathscr{F}$-compatible assignments on $k$ points.  The property that this family forms a $(j,j+1)$-strategy can be expressed as a first order sentence $\xi$.  When $\CSP(\mathbb{A})$ has bounded width (so that NO instances do not have $(j,j+1)$-strategies) the sentence $\xi$ must fail on instances satisfying the $N_{\CSP}$ promise, and must hold on those satisfying the $Y_{(k,\mathscr{F}), \mathsf{Q}}$ promise.

Now assume that $\mathbb{A}$ does not have bounded width. Following Theorem \ref{thm:unaryred} we may assume throughout that the signature of $\mathbb{A}$ contains all singleton unary relations.   If $\mathbb{A}$ fails to have a WNU, then the ANT~\ref{thm:main} applies (as $\texttt{Mod}_p(\texttt{L})\subseteq \texttt{NP}$).    Otherwise, it is known that a direct analogue of the arguments of Section \ref{sec:proof} lead back to a structure $\mathbb{C}$ whose CSP is $\texttt{Mod}_p(\texttt{L})$-complete: indeed $\mathbb{C}$ can be chosen as the template for systems of ternary linear equations over an abelian group structure ${\bf C}$ on the universe $C$; see proof of \cite[Theorem 4.1]{lartes}.   
We then give an analogue to the construction in Section~\ref{sec:3k3SAT}, by reducing the standard YES/NO decision problem of linear equations of size $3$ over ${\bf C}$ to the $(Y_{(2,\varnothing)},N_{\CSP})$ promise problem for linear equations of size $9$.  Here the trick is to replace each variable $x$ by a sum $x_1+x_2+x_3$.  A series of carefully selected regroupings of these equations simplifies the equation length $9$ back down to $3$; each of these simplifications are pp-reductions and careful application of claw formul{\ae} reminiscent of the analysis in Section \ref{sec:3SAT} is required in order to achieve just $(1,\mathscr{F})$-robustness on the YES promise.  This same inflation from $3$-ary equations to $9$-ary equations and then back to $3$-ary equations is repeated inductively, with the achieved level of robust satisfiability growing  exponentially. 
\section{Discussion and extensions}\label{sec:discuss}
We have shown in the ANT~\ref{thm:main} that the fundamental intractability result of~\cite{BJK} can be replaced by an unbounded hierarchy of intractable promise problems, and demonstrated in Theorem~\ref{thm:gap} a collapse in several intermediate complexity classes for these problems.  We feel these results are just the beginning of exciting new applications to ideas relating to the detection of implied constraints (including excluded constraints) as in \cite{beacul}, minimal networks, as in \cite{got}, as well as to other areas of mathematics and computer science, such as the quantum-theoretic applications in \cite{AGK} and the semigroup-theoretic applications of \cite{jac:SAT}.

Some specific new directions this work should be taken include the extension of ANT~\ref{thm:main} to noncore templates and to infinite templates, where a much wider array of important computational problems can be found.  Another difficult question: can the promise supplied by the PCP Theorem be added as a restriction to $N_{\CSP}$ in the ANT~\ref{thm:main}? (We write $N_{\varepsilon\CSP}$ for this condition:  $\varepsilon$ proportion of the constraints must fail.)  The answer is nearly yes, but not quite.  It is quite routine to carry through the failure of a positive fraction of constraints through steps 1--5 of the proof of the ANT~\ref{thm:main}(2), and through step 6 with more difficulty, thereby achieving the \texttt{NP}-completeness of $(Y_{(k,\mathscr{F})},N_{\varepsilon\CSP})$ for core templates without a WNU.  \ Surprisingly though, $N_{\varepsilon\CSP}$ does not in general survive reflection, as the following example demonstrates.  Let $\mathbbold{2}^+$ denote the template on $\{0,1\}$ with the fundamental ternary relation $r$ of +1-in-3$\SAT$ and the $4$-ary total relation $s:=\{0,1\}^4$.  This has no WNU, as +1-in-3$\SAT$ has no WNU, so the  ANT~\ref{thm:main}(2) and claims just made imply that both $(Y_{(k,\mathscr{F})},N_{\varepsilon\CSP})$ and $(Y_{(k,\mathscr{F}),\mathsf{Q}},N_{\CSP})$ are \texttt{NP}-complete.  Yet  $(Y_{(k,\mathscr{F}),\mathsf{Q}},N_{\varepsilon\CSP})$ for $\mathbbold{2}^+$ falls into~$\texttt{AC}^0$!  Indeed the first order property $\tau$ stating that $s$ is total must hold on instances without implied constraints, and fail on any large enough instance $\mathbb{B}$ satisfying $N_{\varepsilon\CSP}$: the number of $r$-constraints is at most $|B|^3$ compared to the $|B|^4$-many $s$-constraints  required by $\tau$, and no $s$-constraint can fail into $\mathbbold{2}^+$.    
For +1-in-3$\SAT$ itself we can show that  $(Y_{(k,\mathscr{F}),\mathsf{Q}},N_{\varepsilon\CSP})$ remains \texttt{NP}-complete.  

\newpage

\section{Appendix: $(k,\mathscr{F})$-types}
\begin{proof}[Proof of Lemma \ref{lem:leqk}]
Assume $\mathbb{B}$ is $(k,\mathscr{F})$-robustly satisfiable and let $\nu:\{b_1,\dots,b_\ell\}\to A$ be a partial assignment compatible with $\mathscr{F}{\mid}_\ell$.  Our goal is to show that $\nu$ extends to a full solution.  For this it suffices to show that $\nu$ extends to an $\mathscr{F}$-compatible assignment on $k$ elements (as we are assuming all such assignments extend to full solutions).  So, let $b_{\ell+1},\dots,b_k$ be \emph{any} $k-\ell$ elements of $B$ not in $\{b_1,\dots,b_\ell\}$ and let $\tau(x_1,\dots,x_k)$ be the $(k,\mathscr{F})$-type of $(b_1,\dots,b_k)$.  So $(b_1,\dots,b_\ell)$ witnesses the truth of the formula $\exists x_{\ell+1}\dots\exists x_{k}\ \tau(x_1,\dots,x_k)$,
a formula in~$\mathscr{F}{\mid}_\ell$.  By the assumed $\mathscr{F}{\mid}_\ell$-compatibility of $\nu$, we have that 
\[
\mathbb{A}\models \exists x_{\ell+1}\dots\exists x_{k}\ \tau(\nu(b_1),\dots,\nu(b_\ell),x_{\ell+1},\dots,x_k).
\]
For each $i=\ell+1,\dots,k$, extend $\nu$ to $b_i$ by letting $\nu(b_i)$ be a witness in $A$ for $\exists x_i$ in the satisfaction of the formula.  So $\mathbb{A}\models \tau(\nu(b_1),\dots,\nu(b_k))$.  Hence, by Lemma~\ref{lem:ktype}, we have an $\mathscr{F}$-compatible extension of $\nu$ to $k$ points, as required.
\end{proof}

\section{Appendix: reflections redux}\label{sec:reflection}
Each quasivariety $K$ is a reflective subcategory of the category of all $\mathscr{R}$-structures (with homomorphisms as morphisms); see Maltsev \cite[V.11.3]{mal} (where the term \emph{replica} is used in place of reflection).  This means that with every $\mathscr{R}$-structure $\mathbb{B}$ one may associate a unique (up to isomorphism) structure $\mathbb{B}^\flat\in K$, the \emph{reflection of $\mathbb{B}$ in $K$}, such that there is a surjective homomorphism $\flat\colon\mathbb{B}\to \mathbb{B}^\flat$, such that for every $\mathbb{C}\in K$ if $\phi\colon\mathbb{B}\to \mathbb{C}$ is a homomorphism, then there is a (unique) map $\phi^\flat\colon\mathbb{B}^\flat\to \mathbb{C}$ such that $\phi^\flat\circ\flat=\phi$.  

In general the construction of $\mathbb{B}^\flat$ from $\mathbb{B}$ is computationally challenging.  In this section we introduce an intermediate notion: the $(k,\mathscr{F})$-reflection, which is the natural concept associated with the construction in Definition \ref{defn:ref1}.  The key ideas behind $(k,\mathscr{F})$-reflection are as follows.
\begin{itemize}
\item There is a first order query such that if $\mathbb{B}$ is $(k,\mathscr{F})$-robustly satisfiable, then $\phi(\mathbb{B})$ is the $(k,\mathscr{F})$-reflection of $\mathbb{B}$ and is $(k,\mathscr{F})$-robustly satisfiable; if $\mathbb{B}$ is a NO instance of $\CSP(\mathbb{B})$, then so also is $\phi(\mathbb{B})$.
\item  If $\mathbb{B}$ is $(k,\mathscr{F})$-robustly satisfiable the $(k,\mathscr{F})$-reflection coincides with the actual reflection in $\mathsf{Q}(\mathbb{A})$.  Thus in the special case of a $(k,\mathscr{F})$-robustly satisfiable instance $\mathbb{B}$, the reflection $\mathbb{B}^\flat$ into $\mathsf{Q}(\mathbb{A})$ can be constructed in low complexity.
\end{itemize}
The remainder of this appendix section gives the required details for constructing $\phi$ and establishing these properties, culminating in Lemma \ref{lem:FOreflection}, which is Lemma \ref{lem:reflect} but with the first order reduction made explicit.

Recall that a quasiequation is a universally quantified sentence of the form 
\[
\alpha_1\wedge\dots \wedge \alpha_n\rightarrow \alpha_0
\]
in which each $\alpha_i$ is an atomic formula.  

\begin{defn}\label{defn:kFtheory}
The \emph{$(k,\mathscr{F})_q$-theory} of $\mathbb{A}$ is the set of first order sentences true on~$\mathbb{A}$ that are of the form
\[
\forall x_1\dots\forall x_k\Big(\sigma(x_1,\dots,x_k)\rightarrow \alpha(x_{i_1},\dots,x_{i_p})\Big)
\]
where $\sigma$ is an $(k,\mathscr{F})$-type and $\alpha$ is an atomic formul{\ae} in $\mathscr{R}\cup\{=\}$, of arity $p$ and with $\{i_1,\dots,i_p\}\subseteq \{1,\dots,k\}$.  
\end{defn}
Up to logical equivalence, the $(k,\mathscr{F})_q$-theory of $\mathbb{A}$ is a subset of the quasi-equational theory of $\mathbb{A}$ due to the following basic rearrangement of quantifiers:   
\[
\left((\exists x\ \phi(x))\rightarrow \psi\right)
\leftrightarrow \forall x\left(\phi(x)\rightarrow \psi\right).
\]

The $(k,\mathscr{F})$-reflection concept builds on a notion of \emph{local reflection} developed in Ham~\cite{ham:SAT}.

\begin{defn}[Local unfrozenness]\label{defn:locallyunfrozen}
Let $\mathbb{A}$ and $\mathbb{B}$ be finite $\mathscr{R}$-structures. Let $k>0$ and a finite set $\mathscr{F}$ of pp  formul{\ae} in $\mathscr{R}$.  The following concepts refer to properties of $\mathbb{B}$ relative to~$\mathbb{A}$.
\begin{itemize}
\item A tuple $(b_1,\dots,b_n)\in B^n$ is said to be $(k,\mathscr{F})$-\emph{frozen} to a relation $r$ of arity $n$ in the signature if $(b_1,\dots,b_n)\notin r^\mathbb{B}$ but every $\mathscr{F}$-compatible partial assignment $\nu$ from a $k$-set containing $\{b_1,\dots,b_n\}$ into $\mathbb{A}$ has $(\nu(b_1),\dots,\nu(b_n))\in r^\mathbb{A}$.
\item The relation $r\in\mathscr{R}\cup\{=\}$ is $(k,\mathscr{F})$-unfrozen on $\mathbb{B}$ \up(with respect to $\mathbb{A}$\up) if no tuple outside of $r^\mathbb{B}$ is $(k,\mathscr{F})$-frozen.
\item $\mathbb{B}$ is \emph{completely $(k,\mathscr{F})$-unfrozen} if all relations in $\mathscr{R}\cup\{=\}$ are unfrozen.
\end{itemize}
\end{defn}
Note that Lemma \ref{lem:leqk} implies that in the definition of a $(k,\mathscr{F})$-\emph{frozen} tuple $(b_1,\dots,b_n)\in B^n$, we could equivalently have used $\mathscr{F}{\mid}_{n}$-compatible assignments from $\{b_1,\dots,b_n\}$ in place of $\mathscr{F}$-compatible assignments from $k$-element sets containing  $\{b_1,\dots,b_n\}$.

Definition \ref{defn:locallyunfrozen} implies that $(k,\mathscr{F})$-frozen tuples correspond exactly to implied constraints that can be detected by examining only $\mathscr{F}$-compatible solutions from $k$-element subsets of $\mathbb{B}$.
Thus ``$(k,\mathscr{F})$-frozen'' is a local version of what Beacham~\cite{bea} calls ``frozen-in'': tuples  implied by  solutions for $\mathbb{B}$.  Beacham and Culberson~\cite{beacul} use the phrase ``frozen''  for a concept dual to frozen-in, which in Beacham~\cite{bea} is referred to as ``frozen-out''.  The phrase ``frozen'' is more widely used to refer to variables specifically: a variable $x$ is ``frozen'' if it takes precisely one value in every solution (it is often implicit that there is at least one solution in this case); see Jonsson and Krokhin~\cite{jonkro} for example.  In the language of Beacham~\cite{bea}, the property that ``$x$ is frozen to the value $a$'' is equivalent to the statement that ``the constraint $x\in\{a\}$ is frozen-in''.  For a general YES instance~$\mathbb{B}$ of $\CSP(\mathbb{A})$ it is possible that a frozen-in constraint $x\in\{a\}$ might not be detectable by examining the $\mathscr{F}$-compatible partial assignments.  However, if the instance $\mathbb{B}$ is $(k,\mathscr{F})$-robustly satisfiable, then every $\mathscr{F}$-compatible assignment must assign $x$ to $a$, because all such partial assignments extend to solutions for~$\mathbb{B}$. Thus, that the constraint $x\in\{a\}$ would also be $(k,\mathscr{F})$-frozen.

\begin{lem}\label{lem:inQ}
Let $\mathbb{A}$ and $\mathbb{B}$ finite $\mathscr{R}$-structures.  Assume that $k\geq \arity(\mathscr{R})$ and that~$\mathbb{B}$ is both $(k,\mathscr{F})$-robustly satisfiable into $\mathbb{A}$ and completely $(k,\mathscr{F})$-unfrozen with respect to $\mathbb{A}$. Then $\mathbb{B}\in \mathsf{Q}^+(\mathbb{A})$.
\end{lem}
\begin{proof}
As $\mathbb{B}$ is  $(k,\mathscr{F})$-robustly satisfiable it has at least one solution by definition.  Thus every $k$-set in $B$ has at least one $\mathscr{F}$-compatible assignment into $\mathbb{A}$.
Now observe that the definition of being completely $(k,\mathscr{F})$-unfrozen implies that for every $r\in \mathscr{R}\cup\{=\}$ (of arity $i\leq k$) and every tuple $(b_1,\dots,b_i)\in B^i\backslash r^\mathbb{B}$, that there is an $\mathscr{F}$-compatible assignment from some $k$-element subset containing $\{b_1,\dots,b_i\}$ placing $(b_1,\dots,b_i)$ outside of $r^\mathbb{A}$.  Thus $\mathbb{B}$ has no implied constraints with respect to $\mathbb{A}$, so lies in $\mathsf{Q}(\mathbb{A})$.  Because $\mathbb{B}$ has at least one homomorphism into $\mathbb{A}$ it then lies in $\mathsf{Q}^+(\mathbb{A})$
\end{proof}

We let the \emph{$(k,\mathscr{F})$-equality relation} $=_{(k,\mathscr{F})}$ denote the relation 
\[
\{(b,b')\in B^2\mid b=b'\text{ is $(k,\mathscr{F})$-frozen}\}.
\]  
The relation $=_{(k,\mathscr{F})}$ is always reflexive and symmetric, but not in general transitive.  We let $\equiv_{(k,\mathscr{F})}$ be the transitive closure of $=_{(k,\mathscr{F})}$.

The following definition is a more careful restatement and extension of Definition \ref{defn:ref1}.
\begin{defn}[$(k,\mathscr{F})$-reflection of $\mathbb{B}$ relative to $\mathbb{A}$]\label{defn:reflection}
Let $\mathbb{B}'$ be the structure obtained from $\mathbb{B}$, but with each relation $r^\mathbb{B}$ for $r\in \mathscr{R}$ enlarged to include all $(k,\mathscr{F})$-frozen tuples for $r$.  The structure $\mathbb{B}^\downarrow$ is defined to be the quotient of $\mathbb{B}'$ by  $\equiv_{(k,\mathscr{F})}$.  We refer to $\mathbb{B}^\downarrow$ as the \emph{$1$-step $(k,\mathscr{F})$-reflection of $\mathbb{B}$ relative to~$\mathbb{A}$}.  

The \emph{$(k,\mathscr{F})$-reflection} of $\mathbb{B}$ relative to $\mathbb{A}$ is denoted by $\mathbb{B}^\Downarrow$ and is obtained from~$\mathbb{B}$ by iteratively applying $1$-step $(k,\mathscr{F})$-reflections until there are no $(k,\mathscr{F})$-frozen tuples.  This will occur after at most a polynomial number of applications of ${}^\downarrow$, because the maximal number of non-hyperedges in each relation $r\in \mathcal{R}\cup\{=\}$ is at most~$|B|^{\arity(r)}$.
\end{defn}
In other words: $\mathbb{B}^\downarrow$ is constructed by examining all $\mathscr{F}$-compatible assignments from $k$-element subsets of $B$ and adding any frozen tuples that are discovered, and then identifying equivalence classes of points that are frozen together by some chain of frozen equalities.  This can be performed in logspace and is what is described Definition \ref{defn:ref1}.  Thus the structure $\mathbb{B}^\Downarrow$ is obtained in polynomial time.  In the case where $\mathbb{B}$ is $(k,\mathscr{F})$-robustly satisfiable  we will prove that $\mathbb{B}^\downarrow=\mathbb{B}^\Downarrow$ (see Lemma \ref{lem:fullref} below), which is why we are able to manage with Definition \ref{defn:ref1} in Section \ref{sec:reflection1} (for polynomial time reductions).

Typically we will drop the phrase ``with respect to $\mathbb{A}$'', when the context of $\mathbb{B}$ as an instance of $\CSP(\mathbb{A})$ is already clear.

\begin{lem}\label{lem:equiv}
If $k\geq 2$ and $\mathbb{B}$ is $(k,\mathscr{F})$-robustly satisfiable into $\mathbb{A}$, then $=_{(k,\mathscr{F})}$ is equal to $\equiv_{(k,\mathscr{F})}$.
\end{lem}
\begin{proof}
Assume $\mathbb{B}$ is $(k,\mathscr{F})$-robustly satisfiable into $\mathbb{A}$ and that $b_1=_{(k,\mathscr{F})}b_2=_{(k,\mathscr{F})}b_3$.  To show that $=_{(k,\mathscr{F})}$ is transitive we prove that $b_1=_{(k,\mathscr{F})}b_3$.

Every homomorphism from $\mathbb{B}$ to $\mathbb{A}$ preserves $=_{(k,\mathscr{F})}$ as the relation $=$ of equality.  Thus every solution identifies $b_1$ with $b_3$.  As every $\mathscr{F}$-compatible assignment on $k$ elements extends to a solution, it follows that every $\mathscr{F}{\mid}_2$-compatible assignment on $\{b_1,b_3\}$ identifies $b_1$ and $b_3$.  Hence $b_1=_{(k,\mathscr{F})}b_3$, as required.
\end{proof}
\begin{lem}\label{lem:typekF}
Let $\mathbb{B}$ be $(k,\mathscr{F})$-robustly satisfiable into $\mathbb{A}$.  For $r\in \mathscr{R}$ of arity $j\leq k$ and a tuple $(b_1,\dots,b_j)\in B^j$, the following are equivalent\up:
\begin{itemize}
\item $(b_1,\dots,b_j)\in r^\mathbb{B}$ or $(b_1,\dots,b_j)\in r$ is $(k,\mathscr{F})$-frozen\up;
\item $\mathbb{A}\models \tau^\mathscr{F}_{b_1,\dots,b_j}(x_1,\dots,x_j)\rightarrow (x_1,\dots,x_j)\in r$, where $\tau^\mathscr{F}_{b_1,\dots,b_j}$ is the $(j,\mathscr{F}{\mid}_j)$-type of $(b_1,\dots,b_j)$.
\end{itemize}
\end{lem}
\begin{proof}
This follows from Lemma \ref{lem:ktype}.
\end{proof}
Note that the implication in Lemma \ref{lem:typekF} is in fact a quasi-equation in the $(k,\mathscr{F})_q$-theory of~$\mathbb{A}$.  This gives the following lemma; we omit full details as it is included only to justify the nomenclature.
\begin{lem}
The $(k,\mathscr{F})$-reflection $\mathbb{B}^\Downarrow$ of a structure $\mathbb{B}$ with respect to~$\mathbb{A}$ is the reflection of~$\mathbb{B}$ into the quasivariety defined by the $(k,\mathscr{F})_q$-theory of $\mathbb{A}$.
\end{lem}

\begin{lem}\label{lem:fullref}
Let $k\geq \arity(\mathscr{R}\cup\{=\})$.  If $\mathbb{B}$ is $(k,\mathscr{F})$-robustly satisfiable into $\mathbb{A}$, then $\mathbb{B}^\downarrow=\mathbb{B}^\Downarrow$ and is the reflection of $\mathbb{B}$ into $\mathsf{Q}(\mathbb{A})$.  Moreover, $\mathbb{B}^\Downarrow$ is $(k,\mathscr{F})$-robustly satisfiable.
\end{lem}
\begin{proof}
We first show that if $\mathbb{B}$ is $(k,\mathscr{F})$-robustly satisfiable, then the $1$-step $(k,\mathscr{F})$-reflection of $\mathbb{B}$ is already the $(k,\mathscr{F})$-reflection $\mathbb{B}^\Downarrow$, which is itself $(k,\mathscr{F})$-robustly satisfiable.
  
Let $\mathbb{B}$ be $(k,\mathscr{F})$-robustly satisfiable, and assume for contradiction, that $\mathbb{B}^\downarrow$ is not completely $(k,\mathscr{F})$-unfrozen.  
So there is a tuple $(b_1,\dots,b_\ell)\notin r^{\mathbb{B}^\downarrow}$ (of arity $\ell\leq k$) such that every $\mathscr{F}$-compatible assignment into $\mathbb{A}$ takes $(b_1,\dots,b_\ell)$ inside~$r^\mathbb{A}$.  
Recall that $\mathbb{B}^\downarrow$ is a quotient of~$\mathbb{B}$ by the relation $=_{(k,\mathscr{F})}$, which is an equivalence relation by Lemma \ref{lem:equiv}.  Select any $b_1',\dots,b_\ell'$ from the $=_{(k,\mathscr{F})}$-equivalence classes of $b_1,\dots,b_\ell$ respectively (there is no requirement that $b_i'=b_j'$ just because $b_i=b_j$).  
Now, in $\mathbb{B}$ we must have $(b_1',\dots,b_\ell')\notin r^\mathbb{B}$, because $(b_1,\dots,b_\ell)\notin r^{\mathbb{B}^\downarrow}$ and $(b_1',\dots,b_\ell')$ maps onto $(b_1,\dots,b_\ell)$ under the quotient map.  Moreover, because $(b_1,\dots,b_\ell)\notin r^{\mathbb{B}^\downarrow}$ we must have that $(b_1',\dots,b_\ell')\notin r^\mathbb{B}$ is not $(k,\mathscr{F})$-frozen.  Therefore there is an $\mathscr{F}$-compatible assignment $\nu:\{b_1',\dots,b_\ell'\}\to A$ that maps $(b_1',\dots,b_\ell')$ outside of $r^\mathbb{A}$.  
However, using the $(k,\mathscr{F})$-robustly satisfiability of $\mathbb{B}$, this assignment extends to a homomorphism $\nu^+:\mathbb{B}\to \mathbb{A}$.  
Every homomorphism from $\mathbb{B}$ into $\mathbb{A}$ factors through the $1$-step $(k,\mathscr{F})$-reflection, and when restricted to $\{b_1,\dots,b_\ell\}$, this contradicts the assumption that  every $\mathscr{F}$-compatible assignment into $\mathbb{A}$ takes $(b_1,\dots,b_\ell)$ inside $r^\mathbb{A}$.  
Hence $\mathbb{B}^\downarrow$ is completely $(k,\mathscr{F})$-unfrozen, so is equal to $\mathbb{B}^\Downarrow$.  Lemma \ref{lem:inQ} now immediately shows that $\mathbb{B}^\Downarrow$ is the reflection of $\mathbb{B}$ into $\mathsf{Q}(\mathbb{A})$.  

Finally, we must prove that $\mathbb{B}^\Downarrow$ is $(k,\mathscr{F})$-robustly satisfiable.  Let $\phi$ be an $\mathscr{F}$-compatible partial assignment from some $k$-element subset $D=\{d_1,\dots,d_k\}$ of $\mathbb{B}^\Downarrow$.  For each $d\in D$, let $d'\in B$ be an element that maps to $d$ under the natural quotient map $\nu$ from $\mathbb{B}$ to $\mathbb{B}^\Downarrow$.  Now any pp-formula $\phi(x_1,\dots,x_k)$ in $\mathscr{F}$ that holds at the tuple $(d_1',\dots,d_k')$ also holds for $(d_1,\dots,d_k)$ in $\mathbb{B}^\Downarrow$.  So  $\phi\circ\nu$ is an $\mathscr{F}$-compatible partial assignment from $\{d_1',\dots,d_k'\}$ to $\mathbb{A}$.  Using the $(k,\mathscr{F})$-robust satisfiability of $\mathbb{B}$, we may extend $\phi\circ\nu$ to a homomorphism from $\mathbb{B}$ to~$\mathbb{A}$.  As all homomorphisms from $\mathbb{B}$ to $\mathbb{A}$ factor through the natural map $\nu$, it follows that~$\phi$ extends to a homomorphism from $\mathbb{B}^\Downarrow$ to $\mathbb{A}$, as required.  As $D$ was arbitrary, we have shown that $\mathbb{B}^\Downarrow$ is $(k,\mathscr{F})$-robustly satisfiable.
\end{proof}
It would follow from Lemma \ref{lem:fullref} that when $\mathbb{B}$ is $(k,\mathscr{F})$-robustly satisfiable, then the reflection into $\mathsf{Q}(\mathbb{A})$ can be constructed in logspace.  The following lemma improves this, and is critical in achieving the \texttt{NP}-completeness results with respect to first order reductions.  Note that lemma makes no claim that a general YES instance $\mathbb{B}$ of $\CSP(\mathbb{A})$ will have $\phi(\mathbb{B})$ a YES instance.
\begin{lem}\label{lem:FOreflection}
If $k\geq \arity(\mathscr{R})$ then there is a $1$-ary first order query $\phi$, such that if $\mathbb{B}$ is $(k,\mathscr{F})$-robustly satisfiable with respect to $\mathbb{A}$, then $\phi(\mathbb{B})$ is $(k,\mathscr{F})$-robustly satisfiable with respect to $\mathbb{A}$ and lies in $\mathsf{Q}^+(\mathbb{A})$, while if $\mathbb{B}$ is a NO instance of $\CSP(\mathbb{A})$, then $\phi(\mathbb{B})$ is a NO instance of $\CSP(\mathbb{A})$.
\end{lem}
\begin{proof}
As is usual for first order reductions, we are allowed access to a linear order $<$ on the set $B$.  We define $\mathbb{B}^\Downarrow$ on a subset of $B$.
First define a binary relation $\equiv$ by 
\[
x\equiv y\Leftrightarrow x\approx y\vee \bigvee_{\mathbb{A}\models \sigma(u,v)\rightarrow u\approx v}\sigma(x,y)
\] 
where in the disjunction, $\sigma$ ranges over all $(2,\mathscr{F}{\mid}_2)$-types.
By Lemma \ref{lem:typekF}, the relation $\equiv$ coincides with the relation $=_{(k,\mathscr{F})}$, so that Lemma~\ref{lem:equiv} shows that for  $(k,\mathscr{F})$-robustly satisfiable instances it is an equivalence relation.  If~$\equiv$ is not an equivalence relation then we output some fixed NO instance using a first order query.  From now we proceed assuming that $\equiv$ is an equivalence relation.  Strictly this bifurcation in the output possibilities is to be built into the following definition, but it is simpler to treat them as separate cases.
We define the universe of $\phi(\mathbb{B})$ to be $\{b\in B\mid \forall b'\ b\equiv b'\Rightarrow b\leq b'\}$.  

Now we define the relations.  In analogy to $\equiv$, we first introduce an intermediate predicate. For each $r\in \mathscr{R}$, with arity $n$, define $\bar{r}$ by 
\[
(x_1,\dots,x_n)\in \bar{r}\Leftrightarrow \bigvee_{\mathbb{A}\models\sigma(x_1,\dots,x_n)\rightarrow (x_1,\dots,x_n)\in r}\sigma(x_1,\dots,x_n),
\]
where again $\sigma$ ranges over all $(n,\mathscr{F}{\mid}_n)$-types (noting that $n\leq k$ by assumption).  Lemma \ref{lem:typekF} shows that $(x_1,\dots,x_n)\in \bar{r}$ logically captures the notion that either $(x_1,\dots,x_n)\in r$ or $(x_1,\dots,x_n)\in {r}$ is $(k,\mathscr{F})$-frozen.

Now define the relation $r^{\mathbb{B}^\Downarrow}$ by
\[
(x_1,\dots,x_n)\in r^{\mathbb{B}^\Downarrow}\Leftrightarrow \exists x_1'\dots\exists x_n'\ \bigg(\Big(\bigwedge_{1\leq i\leq n}x_i\equiv x_i'\Big)\wedge (x_1',\dots,x_n')\in \bar{r}\bigg).
\]
Given that we output a fixed NO instance when $\equiv$ is not an equivalence relation, in all other cases we have output the 1-step $(k,\mathscr{F})$-reflection $\mathbb{B}^\downarrow$.  When $\mathbb{B}$ is $(k,\mathscr{F})$-robustly satisfiable we have $\mathbb{B}^\downarrow=\mathbb{B}^\Downarrow$, which lies in $\mathsf{Q}(\mathbb{A})$ and is $(k,\mathscr{F})$-robustly satisfiable with respect to $\mathbb{A}$ by Lemma \ref{lem:fullref}.  When $\mathbb{B}$ is a NO instance, we have that $\mathbb{B}^\downarrow$ is a NO instance, as any homomorphism from $\mathbb{B}^\downarrow$ to $\mathbb{A}$ would yield a homomorphism from $\mathbb{B}$ to $\mathbb{A}$ by way of composition with the quotient map from $\mathbb{B}$ to $\mathbb{B}^\downarrow$.
\end{proof}

\section{Appendix: example for the proof of Theorem \ref{thm:ppGAP}}\label{eg:ppgap}
We here give an example demonstrating the proof of Theorem \ref{thm:ppGAP} using the pp-definability of 3SAT from 1-in-3SAT (in other words: the reduction of 3SAT to 1-in-3SAT) via the formula
\begin{equation}
\exists y_1\exists y_2\exists y_3\exists y_4\ r(\neg x_1, y_1, y_2)\And r(y_2, x_2, y_3)\And r(y_3, y_4,\neg x_3).\tag{$*$}\label{eq:egSAT}
\end{equation}
where $r(x,y,z)$ denotes the 1-in-3 constraint for $x,y,z$.
Technically, when considered as constraint problems, both these problems involve a number of distinct relations, corresponding to each combination of negations in a clause.  We have abused notation slightly, by retaining negations of variables within clauses and hyperedges: we refer to ``open literals'' as well as ``open elements'' (in case they are negated within a given clause).   Assume then that we have transformed a 3SAT instance $\mathbb{B}$ into a 1-in-3SAT instance according to formula \eqref{eq:egSAT} by replacing each clause $(b_1\vee b_2\vee b_3)$ by the ``family''
\[
r(\neg b_1, c_1, c_2), r(c_2, b_2, c_3), r(c_3, c_4,\neg b_3),
\]
 where $c_1,c_2,c_3,c_4$ are $\exists$-elements introduced specifically for the clause $(b_1\vee b_2\vee b_3)$.  

We consider the case of $k=4$ and note that $\arity(3\SAT)=3$.  Assume that for some~$\mathscr{F}$, the structure $\mathbb{B}$ is a $({\leq}12,\mathscr{F})$-robustly satisfiable 3SAT instance.  Assume that we have been given a partial assignment $\nu$ on a subset $D=\{b,b',c,c'\}$ of $\mathbb{B}^\sharp$, where $b,b'$ are open literals and $c,c'$ are $\exists$-elements.  The goal is to carefully identify at most $|D_\exists|\times 3$ many new open elements $O_B$, find an $\mathscr{F}$-compatible solution to the $\leq 12$ elements in $O_B\cup D_B$, extend it using the $({\leq}12,\mathscr{F})$-robust satisfiability of $\mathbb{B}$ and argue that the extension of this solution to $\mathbb{B}^\sharp$ agrees with the values already given to $D_\exists$.  

To make the example concrete (but not too complicated), assume further that $c,c'$ happen to lie in the same  family of 1-in-3$\SAT$ hyperedges and correspond to the quantified variables $y_2$ and $y_4$ in \eqref{eq:egSAT}, and the open element $b$ has an occurrence in this family (say, in the position $x_1$), but $b'$ does not.    Thus there are open elements $b_2,b_3$ and $\exists$-elements $c_1,c_3$ such that 
\[
r(\neg b, c_1, c), \ r(c, b_2, c_3) \text{ and }r(c_3, c', \neg b_3)
\] 
are hyperedges in $\mathbb{B}^\sharp$.  The set $O_B$ then equals $\{b,b_2,b_3\}$, so that 
$O_B\cup D_B=\{b,b',b_2,b_3\}$.
Let $\sigma(z_1,z_2,z_3,z_4)$ be the $(4,\mathscr{F}{\mid}_4)$-type of the tuple $(b,b_2,b_3,b')$.  Then the following pp-formula
\[
\exists y_1\exists x_2\exists y_3\exists x_3\ r(\neg v_1, y_1, v_2)\And r(v_2, x_2, y_3)\And r(y_3, v_3, \neg x_3)\And \sigma^{r}(v_1,x_2,x_3,v_4)
\]
 is a claw formula in $\mathscr{G}$ and is satisfied by $\mathbb{B}$ at the tuple $(v_1,v_2,v_3,v_4)=(b,c,c',b')$.  If we  assume that $\nu$ is $\mathscr{G}$-compatible then there exist witnesses to $y_1,x_2,y_3,x_3$ such that 
\begin{multline*}
\exists y_1\exists x_2\exists y_3\exists x_3\ r\big(\neg \nu(b), y_1, \nu(c)\big){\And}  r\big(\nu(c), x_2, y_3\big)
\\
{\And} r\big(y_3, \nu(c'), \neg x_3\big){\And} \sigma\big(\nu(b),x_2,x_3,\nu(b')\big)
\end{multline*}
is true in the template for 1-in-3$\SAT$.   Extend $\nu$ from $D$ to include $c_1,b_2,c_3,b_3$ by giving them the chosen witnesses for $y_1$, $x_2$, $y_3$ and $x_3$ respectively.  Then 
\begin{multline*}
r\big(\neg \nu(b), \nu(c_1), \nu(c)\big)\& r\big(\nu(c), \nu(b_2), \nu(c_3)\big)
\& r\big(\nu(c_3), \nu(c'), \neg \nu(b_3)\big)\\\& \sigma^\rho\big(\nu(b),\nu(b_2),\nu(b_3),\nu(b')\big)
\end{multline*}
is true also.  Then it follows that $\nu{\mid}_{\{b,b_2,b_3,b'\}}$ has preserved the $\mathscr{F}$-type $\sigma$ of $(b,b_2,b_3,b')$, hence is $\mathscr{F}$-compatible.  Hence $\nu{\mid}_{\{b,b_2,b_3,b'\}}$ extends to a solution to $\mathbb{B}$, and then $\nu$ extends to $\mathbb{B}^\sharp$ in a way compatible with the given values $\nu(c)$ and $\nu(c')$ (using the already chosen witness values for $c_1,c_3$ to satisfy the clause family containing $c,c'$).

\section{Appendix: Step 3 as a first order reduction}\label{appsec:3k3SAT}
We being by recalling Gottlob's reduction \cite{got}, and show that it can be treated as a first order reduction.  Let $\mathbb{B}$ be an instance of 3SAT over the set of variables $B=\{b_1,\dots,b_n\}$.  For any $k$, consider the set $B_k=\{b_{i,j}\mid i\in\{1,\dots,n\} \text{ and }j\in\{1,\dots,2k+1\}\}$ and construct an instance of $(3k+3)$SAT as follows: each clause $(\bar{b}\vee \bar{c}\vee \bar{d})$ in $B$ (where $\bar{b},\bar{c},\bar{d}$ are elements of $\{b_1,\dots,b_n,\neg b_1,\dots,\neg b_n\}$) is replaced by the family of all $\binom{2k+1}{k+1}^3$ clauses of the form
\[
(\bar{b}_{i_1}\vee\dots \vee \bar{b}_{i_{k+1}}\vee \bar{c}_{i_1'}\vee\dots \vee \bar{c}_{i_{k+1}'}\vee \bar{d}_{i_1''}\vee\dots \vee \bar{d}_{i_{k+1}''})
\]
where $\{i_1,\dots,i_{k+1}\}$, $\{i_1',\dots,i_{k+1}'\}$, $\{i_1'',\dots,i_{k_{k+1}}''\}$ are $(k+1)$-element subsets of $\{1,\dots,2k+1\}$.  Here if $\bar{b}=b_i$ then by $\bar{b}_{i_j}$ we mean $b_{i,i_j}$ and similarly for $\bar{c}$ and $\bar{d}$.  And if $\bar{b}=\neg b_i$ then by $\bar{b}_{i_j}$ we mean $\neg b_{i,i_j}$, and so on.    The corresponding instance of $(3k+3)$-SAT will be denoted by $\mathbb{B}^\sharp$.  It is easy to see that this reduction can be performed in linear time (albeit with a fairly large constant) and in logspace.  We now observe that it can be constructed as a first order query with parameters.

Let the universe of $B^\sharp$ be the subset of the cartesian power $B^{2k+1}$ defined by the $2$-parameter formula $\lambda(x_1,\dots,x_{2k+1},c,d)$ asserting that a tuple is constant except for the $i^{\rm th}$ position (for some $i$) at which point it equals $c$ (or $d$ if the other $2k$ values are already $c$):
\[
\bigvee_{i\in\{1,\dots,2k+1\}}\left(\bigand_{j,j'\neq i} x_j=x_{j}'\And \left((x_j=c\rightarrow x_i=d)\And (x_j\neq c\rightarrow x_i=c)\right)\right)
\]
Notice that for any element $\underline{x}=x_1,\dots,x_{2k+1}$ of this universe, we can apply first order properties to the near unanimity element by referring to $x_1$ if $x_1=x_2$ and $x_3$ otherwise (noting that $2k+1\geq 3$ in the nontrivial case of $k\geq 1$).

Then we define the $(3k+3)$-ary relations of $(3k+3)\SAT$ on $B^\sharp$ (as $(3k+3)\times (2k+1)$-ary relations on $B$).  We use $\underline{x}_i$ to denote a $(2k+1)$ tuple $x_{i,1},\dots,x_{i,2k+1}$, which we take as implicitly satisfying $\lambda$.  Then we may assert that 
$\underline{x}_1,\dots,\underline{x}_{(3k+3)}$ are related to some clause structure if and only if all of $\underline{x}_1, \dots,\underline{x}_{k+1}$ have the same near unanimity value $x$, all of  $\underline{x}_{k+2}, \dots,\underline{x}_{2k+2}$ have the same near unanimity value $y$ and all of $\underline{x}_{2k+3}, \dots,\underline{x}_{3k+3}$ have the same near unanimity value $z$, and that $(x,y,z)$ are related by the corresponding clause structure in $\mathbb{B}$.

For later reference we state the Gottlob's result (and the first order construction here) as a theorem.
The following lemma follows quite easily from the above construction (see Gottlob \cite{got}, using the reduction shown above to be first order).
\begin{thm}\label{thm:got}
The following promise problem is \texttt{NP}-complete \up(for an instance $I$ of $(3k+3)\SAT$\up) with respect to first order reductions\up:
\begin{description}
\item[Yes] $I$ is a $(k,\varnothing)$-robustly satisfiable instance of $(3k+3)\SAT$: every assignment on $k$ variables extends to a solution.
\item[No] $I$ is a NO instance of $(3k+3)\SAT$.
\end{description}
\end{thm}

\section{Appendix: proof of Theorem \ref{thm:3SAT}}
We first give some further details on the selection of values between arrows.  For a $(3k+3)\SAT$ clause $c$ in $\mathbb{B}$, let $\family(c)$ denote the set of $3k+1$ distinct 3$\SAT$ clauses that replace it in $\mathbb{B}^\sharp$: a clause family. The following is asserted without justification in Section \ref{sec:3SAT}.
\begin{lem}\label{lem:arrowbracket}
An assignment $\nu$ satisfies a clause family if and only if between every pair of convergent arrows there is an open element assigned $1$ by $\nu$.  This in turn is equivalent to the property that every  pair of \emph{consecutive} convergent arrows brackets  an open element assigned $1$.
\end{lem}
\begin{proof}
The head of an arrow indicates a $1$ needs to be found in the clause in the direction pointed.  The tail of an arrow $\mapsto$ or $\mapsfrom$ indicates a $1$ has been contributed to the clause pointed from.  Thus $\nu$ does not need to assign $1$ to any open element in a clause if it is the source of an arrow.   This leaves only pairs of consecutive convergent arrows, and here $\nu$ must assign an open element in such a clause the value~$1$.
\end{proof}

The literals lying between to convergent arrows will be refered to as an \emph{arrow interval}.  A pair of convergent arrows (or an arrow interval) will be said to have been \emph{stabilised} by a partial assignment if some open literal lying in the interval has been assigned the value $1$.

We mention that there is always a simpler family of formulas that could replace~$\mathscr{F}$  in the argument of Section \ref{sec:3SAT}.  In particular, instead of the talon $\gamma$ consisting of conjunctions of entire clause familes, we only need to include that part of the clause family that covers minimal unstabilised arrow intervals.  Such intervals are at worst one less in size than the entire clause family.  Even then there is redundancy.  As an example, we prove the result sketched in Abramsky, Gottlob and Kolaitis \cite{AGK}, where it is stated that $(3,\Gamma_{3\SAT})$-robustness is \texttt{NP}-complete (where $\Gamma_{3\SAT}$ denotes the quantifier free atomic formul{\ae} of the $3\SAT$ relations), by reduction from $(3,\varnothing)$-robustness of $12\SAT$.  In fact we observe that there is a technicality to the reduction, without which the argument fails.
\begin{thm}\label{thm:3clean}
The following promise problem is \texttt{NP}-complete.
\begin{description}
\item[YES] $\mathbb{B}$ is a $(3,\mathscr{F})$-robustly satisfiable instance of $3\SAT$, where $\mathscr{F}$ consists only of the fundamental relations of $3\SAT$
\item[NO] $\mathbb{B}$ is a NO instance of $3\SAT$
\end{description}
\end{thm}
\begin{proof}
We use the previous argument, except that we take more care with the set~$\mathscr{F}$.  Assume we are given a partial assignment $\nu$ on the constructed $3\SAT$ instance $\mathbb{B}^\sharp$, and let $D_B$ and $D_\exists$ be as before.

If $|D_B|=3$ then there is nothing to do, because we are assuming $\mathbb{B}$ is $3$-robustly satisfiable.  If $|D_B|=2$, then the single $\exists$-element $c$ assigned by $\nu$ creates exactly one arrow interval (namely with one end of its family of clauses), which always has at least two distinct open elements in it.  If these are the variables in $D_B$ then local compatibility with $3\SAT$ ensures that $c$ is already stabilized and there is nothing to do.  Otherwise, there is an open element in the arrow interval created by $\nu(c)$ that does not appear in $D_B$.  Add this element to $D_B$ and assign it the value that stabilises the arrow interval of $c$.

Now assume there are two variables $c,c'\in D_\exists$ and one open element $b\in D_B$.  The variables $c,c'$ can create either one minimal arrow interval (if they appear in a common family of clauses, and they direct toward each other), or two arrow intervals (against endpoints of families of clauses).  In the first option, locally compatibility with 3\SAT\ clauses ensures that either $\nu(b)$ stabilises the arrow interval between $c$ and $c'$, or there is a open element $b'\in B\backslash\{b\}$ that lies in this interval, and can be used to stabilise the arrow interval.  

In the case where $c$ and $c'$ create two arrow intervals, then these intervals each contain at least two open elements, say $u_c,v_c$ and $u_{c'},v_{c'}$.  It is possibly however that $\{u_c,v_c\}$ and $\{u_{c'},v_{c'}\}$ are not disjoint, and also that $b$ could be in one or both of these sets.  Provided that we can find an unallocated open literal in both these arrow intervals, the values of $c,c'$ can be stabilised.  This is obviously possible in the case except when  $\{u_c,v_c\}=\{u_{c'},v_{c'}\}$ and contains $b$.  Indeed, it is not possible to extend the following assignment:
$b,c,c'\mapsto 0$ for the pair of clauses $(b\vee u\vee c), (b\vee \neg u\vee c')$.  Thus to recover the result claimed in \cite{AGK} we must show that this does not happen for the particular family of clauses we have constructed.  Indeed, if we examine the clauses of $(3\times 3+3)\SAT$ constructed in Section \ref{sec:3k3SAT}, it can be seen that all literals descending from one of the original variables in $3\SAT$ either appear all negated in a clause, or all nonnegated.  In the one problematic case observed here, we are considering open literals occurring at the extreme end of the family of clauses.  Thus we cannot have $b$ positive and $u$ negative in the same clause.  

Finally we consider the case where $|D_\exists|=3$, say $D_\exists=\{c,c',c''\}$.  In this case we get either one, two or three minimal arrow intervals.  One arrow interval is stabilised by selecting just one open element in the interval. In the case of two arrow intervals, observe that at least one must be against the endpoints of a family, hence include at least two open elements.  The the two intervals can be stabilised by selecting at most two open elements.  For the case of three intervals: observe that each of these intervals contains two open elements.  Even in the worst case, where there are in total only two open elements appearing in these intervals, there is always an assignment on at most $3$ elements that stabilise the values of $c,c',c''$.  In the case where two open elements $b,b'$ appear throughout, observe that there are four possible truth assignments to $b,b'$, yet only at most $3$ combinations of negation and non-negation present in the intervals required to stabilise.  Hence one combination is not present.  By fixing an assignment that fails this combination, all of the other three combinations are satisfied, hence stabilise $c,c',c''$.
\end{proof}
\begin{rem}\label{rem:24}
A very similar analysis when $k=2$ shows that $(Y_{(2,\varnothing)},N_{\CSP})$ is \texttt{NP}-complete for $3\SAT$.  In the case of $k=4$, only a little more effort shows that it is sufficient for $\mathscr{F}$ to include just the fundamental relations of $3\SAT$ and the one extra formula
\[
\exists y\ (x_1\vee y\vee x_2)\wedge (x_3\vee \neg y\vee x_4)
\]
(along with the various negations of variables and arrangements of literals within the clauses: so, each of $x_1,\dots,x_5$ are literals, possibly negated).
Note that this formula eliminates the tuple $(0,0,0,0)$ as a viable assignment for $(x_1,x_2,x_3,x_4)$.
\end{rem}
\section{Appendix: distillation and a problem of Beacham and Culberson}
We here make a slight diversion from the details of proofs to observe how to use the technique of the previous section to resolve one line of investigation instigated by Beacham and Culberson in \cite{beacul}.   
In \cite{beacul} it is shown that $(1,\varnothing)$-robust satisfiability is \texttt{NP}-complete for $n\SAT$ (when $n>2$) and that $(n,\varnothing)$-robust satisfiability is trivial: indeed $(n,\varnothing)$-robust satisfiability holds only for instances with no clauses, so is solvable in $\texttt{AC}^0$.  At what value of $k$ in $1\leq k\leq n$ does the complexity of $(k,\varnothing)$-robust satisfiability for $n\SAT$ transition downward from \texttt{NP}-complete?  Such questions are invited in \cite{beacul}, and we answer this by showing that the boundary is at $k=n-1$ and moreover holds at the level of the promise problem $(Y_{(k,\varnothing)}, N_{\CSP})$.  The result is not stated in the main body of the article.
\begin{thm}
For any $n>2$ the following promise problem $(Y_{(n-1,\varnothing)}, N_{\CSP})$ is \texttt{NP}-complete for $n\SAT$.\end{thm}
\begin{proof}
As noted in Remark \ref{rem:24}, a very easy adaptation of the proof of Theorem \ref{thm:3clean} shows that the promise $(Y_{(2,\varnothing)},N_{\CSP})$ is \texttt{NP}-complete for $3\SAT$.  For the remainder of the proof we fix $n\geq 4$.  If in Section~\ref{sec:3k3SAT} we begin with $(n-1)\SAT$ instead of $3\SAT$, then the same variable replacement (with $k:=n-1$) detailed in Section~\ref{sec:3k3SAT} will show that the promise $(Y_{(n-1,\varnothing)},N_{\CSP})$ is \texttt{NP}-complete for $(n^2-n)\SAT$ (noting $(n-1)(k+1)=n^2-n$).  The reduction of $(n^2-n)\SAT$ to $n\SAT$ takes each $(n^2-n)$-clauses to a family consisting of $n+1$ distinct $n$-clauses.
We need to show that every assignment from $n-1$ elements extends to a solution.
Note that Lemma \ref{lem:arrowbracket} continues to hold, with the notion of arrow intervals defined for $n$-clauses in identical fashion to that defined for $3$-clauses.  Let $\nu:D\to\{0,1\}$ be any assignment from some $n-1$-element set of elements~$D$.  We can replace any $\exists$-element $c$ in $D$ as follows: there are always at least $n-2$ open elements between any two convergent arrows (and $n-1$ such elements at the extreme ends of clause families). If one of these open elements is not in $D$, then we can replace the $\exists$-element with this open element.  If all $n-2$ of the open elements are in $D$, then, as $|D|=n-1$, the $c$ was the only $\exists$-element in~$D$.  Then, the other arrow making the arrow bracket with $c$ is one placed at the extreme end of the clause family.  In this instance, there are more than $n-2$ distinct open elements in the arrow bracket, so that at least one can be selected outside of $D$, to replace~$c$.
\end{proof}
\begin{cor}
Let  $n>2$.  For  $k\geq n$ the $(k,\varnothing)$-robust satisfiability for $n\SAT$ is in $\texttt{AC}^0$, and for $k<n$ the  problem $(Y_{(k,\varnothing)}, N_{\CSP})$ is 
\texttt{NP}-complete.
\end{cor}
\section{Appendix: proof of Theorem \ref{thm:unaryred}}
We now give this proof in full detail.
\begin{proof}[Proof of Theorem \ref{thm:unaryred}]
Let $\mathbb{B}$ be an instance of $\CSP(\mathbb{A}_{\Con})$ and recall that for $a\in A$, the notation $\underline{a}$ denotes the unary relation which in $\mathbb{A}_{\Con}$ is the singleton set $\{a\}$.  We construct an instance $\mathbb{B}'=\langle B';{\mathcal R}^{\mathbb{B}'}\rangle$ of $\CSP(\mathbb{A})$ in two steps.  We first apply the function $\phi$ of Lemma \ref{lem:fullref}, which is an $\texttt{AC}^0$ reduction from $\CSP(\mathbb{A}_{\Con})$ to itself.   We then apply the construction given in \cite[Theorem 4.7]{BJK}, which reduces $\CSP(\mathbb{A}_{\Con})$ to $\CSP(\mathbb{A})$. Note that this reduction is shown to be first-order in \cite[Lemma $2.5$]{lartes}.  

It remains to show if $\mathbb{B}$ is $(k,\mathscr{F})$-robustly satisfiable with respect to $\mathbb{A}_{\Con}$, then $\mathbb{B}'$ is $(k,\mathscr{G})$-robustly satisfiable with respect to $\mathbb{A}$, for some finite set of pp-formul{\ae} $\mathscr{G}$ in the language of $\mathscr{R}$. 

Assume $\mathbb{B}$ is $(k,\mathscr{F})$-robustly satisfiable; then so too is the structure $\phi(\mathbb{B})$ of Lemma~\ref{lem:fullref}, so from this point on, there is no loss of generality in assuming  that~$\mathbb{B}$ is both $(k,\mathscr{F})$-robustly satisfiable and $(k,\mathscr{F})$-reflected.  We will also assume that~$\mathscr{F}$ contains the singleton fundamental relations of $\mathcal{R}_{\Con}$, as otherwise we may add them. 

We now recall the construction of $\mathbb{B}'$. Let $C_A=\{c_a\mid a\in A\}$ be a copy of the set $A$.  The universe $B'$ is the  union $B\cup C_A$, which we initially assume to be disjoint.  The relations $\mathcal R^{\mathbb{B}'}$ are constructed from ${\mathcal R}^{\mathbb{B}}_{\Con}$ by removing every hyperedge $b\in \underline{a}$ in $\mathcal R^{\mathbb{B}}_{\Con}$ but identifying $b$ with $c_a$, and then including, for all $r\in \mathscr{R}$, all hyperedges $(c_{a_1},\dots,c_{a_{\arity(r)}})\in r^{B'}$, where $(a_1,\dots,a_j)\in r^A$.  Note that because $\mathbb{B}$ is $(k,\mathscr{F})$-reflected it follows that distinct elements of $B$ cannot be identified with the same element of $C_A$, thus both $\mathbb{B}$ and $\mathbb{A}$ are (isomorphic to) induced substructures of $\mathbb{B}'$.  Let $B_1$ denote $B\backslash C_A$ (those elements of $B$ not identified with any element of $C_A$), let $B_2$ denote $B\cap V_A$ (those elements identified with an element of $C_A$) and let $B_3$ denote $C_A\backslash B$ (those elements of~$C_A$ not identified with any element of $B$).

Now we construct the formul{\ae} in $\mathscr{G}$.  In all possible combinations, select 
$p\leq k$ with $k-p\leq |A|$, a $(p,\mathscr{F})$-type $\sigma(x_1,\dots,x_p)$ and a $(k-p)$-element subset  
\[\{x_{a_{i_{p+1}}},\dots,x_{a_{i_{k}}}\}\subseteq \{x_{a_1},\dots,x_{a_{|A|}}\}, 
\]
where $\{x_{a_1},\dots,x_{a_{|A|}}\}$ is a copy of the universe $A$.  
For each such combination of selections construct a pp-formula $\delta(x_1,\dots,x_p,x_{a_{i_{p+1}}},\dots,x_{a_{i_{k}}})$ of arity  $k$ by taking the conjunction 
\[
\diag_{\mathbb A}(x_{a_1}, \dots, x_{a_{|A|}})\wedge \sigma(x_1,\dots,x_p),
\]
replacing all formul{\ae} of the form $x\in \underline{a}$ in $\sigma$ by $x=x_{a}$, and then existentially quantifying all variables not in $\{x_1,\dots,x_p,x_{a_{i_{p+1}}},\dots,x_{a_{i_{k}}}\}$.  Note that in this construction, if $x_i\in\{x_1,\dots,x_p\}$ is constrained to $\underline{a}$ in $\sigma$ and $a\notin \{a_{i_{p+1}},\dots,a_{i_k}\}$, then the conjunct $x_i=x_a$ appears in $\delta$, with $x_i$ unquantified and $x_a$  quantified.  Of course, if $\delta$ is true at some tuple, then any witness to the quantified variable $x_a$ must coincide with the value given to $x$.  The set $\mathscr{G}$ consists of all  formul{\ae} $\delta$ constructed in this way.

Let $D=\{b_1,\dots, b_k\}$ be a $k$-element subset of $B'$ and consider a $\mathscr{G}$-compatible assignment $\alpha\colon D \to A$.
Arrange the elements of $\{b_1,\dots, b_k\}$ so that $b_1, \dots, b_p\in B_1\cup B_2$ and $b_{p+1},\dots,b_{k}\in B_3$.  Note that $\{b_{p+1},\dots,b_k\}$ is a $(k-p)$-element subset $\{c_{a_{i_{p+1}}},\dots,c_{a_{i_{k}}}\}$ of $\{c_{a_1},\dots,c_{a_{|A|}}\}$.  Let $\sigma(x_1,\dots,x_p)$ denote the $(p,\mathscr{F}{\mid}_p)$-type of $(b_1,\dots,b_p)$ in $\mathbb{B}$.   We find that the formula  $\delta(x_1,\dots, x_p,x_k,x_{a_{i_{p+1}}},\dots,x_{a_{i_{k}}})$ is in $\mathscr{G}$, and $\mathbb{B}'\models \delta(b_1,\dots, b_k)$.   Then $\mathscr{G}$-compatibility ensures $\delta$ is preserved by $\alpha$ and hence $\mathbb{A}$ satisfies $\delta(\alpha(b_1), \dots, \alpha(b_k))$. Thus, there is a tuple of witnesses to the quantified variables in $\delta$. 
Let us extend $\alpha$ to cover all of the elements in~$C_A$, by letting $\alpha(c_a)$, for each $c_a\notin \{b_{p+1},\dots,b_k\}$, be the witness to the corresponding existentially quantified variable $x_a$ in satisfaction of $\delta(\alpha(b_1), \dots \alpha(b_k))$.  Note that for $b_i$ in $D\cap B_2$, there is $a\in A$ such that the constraint $b_i\in \underline{a}$ was in~$\mathcal{R}_{\Con}^\mathbb{B}$ (so that $b_i$ is the element $c_a$ of $B'$).  Because all unary relations appear in~$\mathscr{F}$, the type $\sigma(x_1,\dots,x_p)$ of $(b_1,\dots,b_p)$ includes the conjunct $x_i\in \underline{a}$, so that $\delta(\alpha(b_1), \dots \alpha(b_k))$ includes the conjunct $\alpha(b_i)=x_a$.  Thus the selected witness to $x_a$ must coincide with $\alpha(b_i)$.

Now the induced substructure of $\mathbb{B}'$ on $B_2\cup B_3=C_A$ is isomorphic to $\mathbb{A}$.  Hence the restriction $\alpha\rest{B_2\cup B_3}$ is an automorphism~$\beta$ of $\mathbb{A}$.  Because $\beta^{-1}$ preserves the pp-definable relations of $\mathbb{A}$, it follows that $\mathbb A$ satisfies $\delta(\beta^{-1}\circ\alpha(b_1), \dots, \beta^{-1}\circ\alpha(b_k))$ also.  The function $\beta^{-1}\circ\alpha\rest{B_2\cup B_3}$ is the map sending $c_a$ to $a$.  Hence  $\sigma(\beta^{-1}\circ\alpha(b_1),\dots, \beta^{-1}\circ\alpha(b_p))$ is true in $\mathbb{A}_{\Con}$, and because $\sigma$ is the $(p, \mathscr{F}{\mid}_p)$-type of $b_1, \dots, b_p$, it follows that $\beta^{-1}\circ\alpha\rest{\{b_1,\dots,b_p\}}$ is $\mathscr{F}$-compatible for $\mathbb{I}$.  Now $(\leq k, \mathscr{F})$-robustness of $\mathbb{B}$ allows us to extend $\beta^{-1}\circ\alpha\rest{B}$ to a full solution $\varphi$ of~$\mathbb{B}$. We can extend $\varphi$ to a solution $\psi$ of $\mathbb{B}'$ by sending each $c_a$ in $C_A$ to $a$. Then $\beta\circ\psi$ is a solution of $\mathbb{B}'$ that extends the values $\alpha(b_1), \dots, \alpha(b_k)$, as required.
\end{proof}

\section{Appendix: proof of Theorem \ref{thm:gap}(1)}
We now give full details for the proof of Theorem \ref{thm:gap}(1) sketched in Section \ref{sec:BW}.  

\begin{defn}
Let $j\geq 0$, let $\mathbb{A}$ a finite relational structure and $\mathbb{B}$ an input for $\CSP(\mathbb{A})$.  A \emph{$(j,j+1)$-strategy} for $\mathbb{B}$ is a family of partial functions $\mathscr{P}$, with each partial function having domain some subset of $B$ with size at most $j+1$, and codomain the set $A$ and satisfying\up:
\begin{description}
\item[(Homomorphism)] each $f\in\mathscr{P}$ is a homomorphism from the induced substructure of $\mathbb{B}$ on the domain of $f$.
\item[(Restriction)]  if $f\in\mathscr{P}$ then all restrictions of $f$ are in $\mathscr{P}$.
\item[(Extension)] if $f\in\mathscr{P}$ has domain $D$ with $|D|<j+1$ and $d_1,\dots,d_{i}$ for $i\leq j+1-|D|$ are elements of $B\backslash D$, then there is $g\in\mathscr{P}$ with domain $D\cup\{d_1,\dots,d_i\}$ such that $f$ is the restriction of $g$ to $D$.
\end{description}
\end{defn}
\begin{defn}\label{lem:boundedwidth}
Let $j\geq 0$, let $\mathbb{A}$ a finite relational structure and $\mathbb{B}$ an input for $\CSP(\mathbb{A})$.  The problem $\CSP(\mathbb{A})$ has \emph{bounded width} $j+1$ if an instance $\mathbb{B}$ is a YES instance if and only if it has a $(j,j+1)$-strategy.  If in addition, the partial solutions in every $(j,j+1)$-strategy is the restriction of a solution,  then $\CSP(\mathbb{A})$ is said to have \emph{strict width}.
\end{defn}
It is known that $\CSP(\mathbb{A})$ has strict width if and only if it has a near unanimity polymorphism, and has bounded width if and only if its polymorphism algebra generates a congruence meet-semidistributive variety~\cite{barkoz:BW}.  Polymorphism conditions for this last property are stated in Theorem \ref{thm:BW}(2) above.

There should be no surprise that $(k,\mathscr{F})$-robust satisfiability has some relationship to the bounded width $k$ (or less) property.  Indeed if $\mathbb{B}$ is  $(k,\mathscr{F})$-robustly satisfiable then the homomorphisms from $\mathbb{B}$ are determined by $\mathscr{F}$-compatible partial maps from $k$-element domains.  Bounded width $k$ implies also that homomorphisms are determined by special families of of partial homomorphisms from domains of size at most $k$.  The connection is close, but in general problems of bounded width may be as hard as \texttt{P}-complete, so it is perhaps initially a little surprising that the $(Y_{(k,\mathscr{F})},N_{\CSP})$ problem collapses to the class $\texttt{AC}^0$.  The difference is that in the presence of $(k,\mathscr{F})$-robust satisfiability we are able to identify a specific family of partial homomorphisms that must form a $(j,j+1)$-strategy.

For a finite set of pp-formul{\ae} $\mathscr{F}$, we say that a partial map from a subset of the instance~$\mathbb{B}$ with size at most $k$ to the template $\mathbb{A}$ is $\mathscr{F}_{\leq k}$-compatible if it extends to a $\mathscr{F}$-compatible assignment on $k$ elements.
\begin{lem}\label{lem:jstra}
Assume $\mathbb{B}$ is $(k,\mathscr{F})$-robustly satisfiable with respect to $\mathbb{A}$.  Then for any $j\leq k$, the family of all 
$\mathscr{F}_{\leq k}$-compatible solutions on at most $j$ points is a $(j-1,j)$-strategy for $\mathbb{B}$.  
\end{lem}
\begin{proof}  
Trivial from definitions.  We mention that initially the homomorphism condition appears to be difficult to verify, given that there is no restriction on $\mathscr{F}$.  For example, if $\mathscr{F}$ is empty, then an $\mathscr{F}_{\leq k}$-compatible assignment $f$ need not in general preserve the fundamental relations.  But this cannot happen when $\mathbb{B}$ is $(k,\mathscr{F})$-robustly satisfiable, because by assumption such an $f$ must extend to a homomorphism.%
\end{proof}

We now show how the property that ``the $\mathscr{F}_{\leq k}$-compatible solutions on at most $j$ points form a $(j-1,j)$-strategy'' can be expressed as a first order sentence.  There are two ingredients.  The first is that for a tuple $(b_1,\dots,b_i)$ (for $i\leq j+1$), the property that every $\mathscr{F}_{\leq k}$-compatible solution is a homomorphism can be expressed as a first order sentence.  This is trivial if $\mathscr{F}$ contains the fundamental relations of $\mathscr{R}$, but as even the $\mathscr{F}=\varnothing$ case is of interest, we do not want to assume this.  

The second ingredient is that the extension property, expressed in terms of $\mathscr{F}_{\leq k}$-compatibility, can be written as a first order sentence.  The restriction property always holds, by the definition of $\mathscr{F}_{\leq k}$-compatibility.


First, some notation.  Let $\phi(x_1,\dots,x_i)$ be a formula in prenex form, with all quantifiers existential and let $\phi'(x_1,\dots,x_i,z_1,\dots,z_j)$ be the underlying open formula.  Then the assertion that $\phi(x_1,\dots,x_i)$ holds with all witnesses to quantifiers coming from some list $y_1,\dots,y_q$ is the disjunction over all functions $\iota:\{z_1,\dots,z_j\}\to\{y_1,\dots,y_q\}$ of the open formul{\ae}
$\phi'(x_1,\dots,x_i,\iota(z_1),\dots,\iota(z_j))$.
We denote this open formula by $\phi{\mid}_{\{y_1,\dots,y_q\}}(x_1,\dots,x_i)$.
\begin{lem}\label{lem:itypeformula}
Let $\mathbb{B}$ be a finite $\mathscr{R}$-structure and $\mathscr{F}$ a finite set of pp-formul{\ae} in~$\mathscr{R}$.  Let $(b_1,\dots,b_i)$ be a tuple of distinct elements from $B$ and $y_1,\dots,y_q$ be a list of variables.  There is a first order formula with free variables in $y_1,\dots,y_q$ and parameters $b_1,\dots,b_i$, which asserts that the $(i,\mathscr{F}{\mid}_i)$-type of $(b_1,\dots,b_i)$ in $\mathbb{B}$ has witnesses from amongst $y_1,\dots,y_q$.
\end{lem}
\begin{proof}
We first prove this when $i=k$.     In this case, the desired formula  can be found by taking the disjunction over all subsets $S$  of $\mathscr{F}_k$ of the following:
\begin{equation}
\bigwedge_{\phi(x_1,\dots,x_k)\in S}\phi{\mid}_{\{y_1,\dots,y_q\}}(b_1,\dots,b_k)\wedge\bigwedge_{\phi(x_1,\dots,x_k)\in \mathscr{F}_k\backslash S}\neg\phi(b_1,\dots,b_k).\tag{$*$}\label{eqn:stricttype}
\end{equation}
Indeed, for each $S$, the first block of conjuncts assert that the formul{\ae} in $S$ are amongst the type of $(b_1,\dots,b_k)$ (with witnesses chosen from the required elements), while the second block of conjuncts assert that no other formula in~$\mathscr{F}_k$ holds at $(b_1,\dots,b_k)$.
Now we consider $i< k$.  Recall that $\mathscr{F}{\mid}_i$ consists of all formul{\ae} $\exists x_{i+1}\dots\exists x_k\ \sigma(x_1,\dots,x_k)$, where $\sigma$ is a $(k,\mathscr{F}_k)$-type.  To find the  formula required by the lemma, we cannot use exactly the same method as for $i=k$, as now there are extra quantified variables $x_{i+1},\dots,x_k$ that appear in both the first block of conjuncts of  \eqref{eqn:stricttype} as well as the second block.  We need these to come from $\{y_1,\dots,y_q\}$, but we do not care where the remaining quantified variables in the second block come from.  Therefore, it suffices to replace each disjunct of the form \eqref{eqn:stricttype} by the disjunction over all functions $\iota:\{b_{i+1},\dots,b_k\}\to\{y_1,\dots,y_q\}$ of 
\begin{multline*}
\bigwedge_{\phi(x_1,\dots,x_k)\in S}\phi_{\{y_1,\dots,y_q\}}(b_1,\dots,b_i,\iota(b_{i+1}),\dots,\iota(b_k))\\
\wedge\bigwedge_{\phi(x_1,\dots,x_k)\in \mathscr{F}_k\backslash S}\neg\phi(b_1,\dots,b_i,\iota(b_{i+1}),\dots,\iota(b_k)).
\end{multline*}
\end{proof}
When we use this lemma we will set $q$ large enough to capture enough witnesses to the possible $(i,\mathscr{F}{\mid}_i)$-type of an $i$-tuple.  An upper bound for this  $k+\ell$, where $\ell$ is the maximal number of quantifiers appearing in any $(k,\mathscr{F}_k)$-type: this is because $k-i+\ell$ is the maximum number of quantifiers appearing in an $(i,\mathscr{F}{\mid}_i)$-type (which involve taking an $(k,\mathscr{F}_k)$-type and quantifying a further $k-i$ variables). 
%
%
\begin{lem}\label{lem:FOstrategy}
For any fixed $k\geq 0$ and $\mathbb{F}$ and any $j\leq k$, the class of $\mathscr{R}$-structures for which the 
$\mathscr{F}_{\leq k}$-compatible solutions on at most $j$ points is a $(j-1,j)$-strategy is first order definable.
\end{lem}
\begin{proof}
We need to verify the homomorphism, restriction and extension properties of a $(j-1,j)$-strategy.   Let $q$ be $k+\ell$, where $\ell$ is the maximal number of quantifiers appearing in any $(k,\mathscr{F}_k)$-type.

{\bf Homomorphism.} When $\mathscr{F}$ includes the fundamental relations of $\mathscr{R}$ then there is nothing to do, as compatibility of a partial function with the fundamental relations is the definition of a partial homomorphism.  When $\mathscr{F}$ does not contain $\mathscr{R}$ then we need to work harder.  We use the fact that the type of a tuple $(b_1,\dots,b_j)$ captures $\mathscr{F}$ compatibility, and that at most a further $q$ elements $c_1,\dots,c_q$ are required to witness the $(j,\mathscr{F}{\mid}_j)$-type of a $j$-tuple.  When such a set $\{c_1,\dots,c_q\}$ of witnesses have been identified, it suffices to know that on the induced substructure on $\{b_1,\dots,b_j,c_1,\dots,c_q\}$, all $\mathscr{F}$-compatible assignments from $\{b_1,\dots,b_j\}$ are partial homomorphisms.  This can be achieved as a first order sentence because there are only finitely many $\mathscr{R}$-structures on at most $j+q$ elements, and so there exists a finite list consisting of all  $\mathscr{R}$-structures on $\{b_1,\dots,b_j,c_1,\dots,c_q\}$ for which all $\mathscr{F}$-compatible assignments from $\{b_1,\dots,b_j\}$ into $\mathbb{A}$ are partial homomorphisms.  

For each $i\leq j$, let $J_{i}$ be the set of all $\mathscr{R}$-structures $\mathbb{D}$ on $p:=q+j$ elements $x_1,\dots,x_{p}$, for which all $\mathscr{F}$-compatible assignments into $\mathbb{A}$ from points $x_1,\dots,x_i$ are homomorphisms from the induced substructure of $\mathbb{D}$ on $\{x_1,\dots,x_i\}$.    Note that $J_{i}$ will typically contain a number of isomorphic copies of its members, but is a finite set, as $p$ is fixed and the universe is a fixed set.  
For each $\mathbb{D}\in J_{i}$, let $\delta_{\mathbb{D}}(x_1,\dots,x_p)$ be the diagram of $\mathbb{D}$: in other words, the open formul{\ae} consisting of the conjunction of all atomic and negated atomic formul{\ae} true in $\mathbb{D}$.  Let $\alpha_{J,i}$ be the disjunction of the $\delta_{\mathbb{D}}$ over all $\mathbb{D}$ in $J_{i}$.  

By calling on Lemma \ref{lem:itypeformula} the following can be written as a first order sentence $\beta_{J,i}$ stating that every $\mathscr{F}$-compatible assignment from $\mathbb{B}$ on $i$ points is a homomorphism is the following:
\begin{quotation}
\noindent``for all pairwise distinct $x_1,\dots,x_i,x_{i+1},\dots,x_{i+q}$, if the set $\{x_{i+1},\dots,x_{i+q}\}$ contains all necessary witnesses to the satisfaction of the $(i,\mathscr{F}{\mid}_i)$-type of $x_1,\dots,x_i$, then $\alpha_{J,i}(x_1,\dots,x_p)$ holds.''
\end{quotation}
The conjunction $\beta_{\rm hom}:=\bigwedge_{0\leq i\leq j}\beta_{J,i}$ states that all $\mathscr{F}_{\leq k}$-compatible assignments on at most $j$ points are homomorphisms from their domain.

{\bf Restriction.}  In the present setting, this follows immediately from the definitions, because the restriction of an $\mathscr{F}$-compatible assignment is also $\mathscr{F}$-compatible.  

{\bf Extension.} For each $i<i'\leq j$, let $K_{i,i'}$ be the set of all $\mathscr{R}$-structures on $p$ elements $x_1,\dots,x_{p}$ for which every $\mathscr{F}{\mid}_{i}$-compatible assignment from $x_1,\dots,x_i$ into $\mathbb{A}$ extends to an $\mathscr{F}{\mid}_{i'}$-compatible assignment from $x_1,\dots,x_{i'}$ into $\mathbb{A}$.   Note that $K_{i,i'}$ is a finite set of structures.  Let  $\alpha_{K,i,i'}$ be the disjunction of the diagrams $\delta_{\mathbb{D}}$ over all $\mathbb{D}$ in $K_{i,i'}$.  Then the following sentence $\beta_{K,i,i'}$ asserts that every $\mathscr{F}{\mid}_i$-compatible assignment on $i$ elements of $\mathbb{B}$ into $\mathbb{A}$ extends to an $\mathscr{F}{\mid}_{i'}$-compatible assignment on any $i'-i$ further elements.
\begin{quotation}
``for all pairwise distinct $x_1,\dots,x_{p}$, if the set $x_{i'+1},\dots,x_{p}$ contains the witnesses to satisfaction of both the $(i,\mathscr{F}{\mid}_i)$-type of $(x_1,\dots,x_i)$ and $(i',\mathscr{F}{\mid}_{i'})$-type of $(x_1,\dots,x_{i'})$, then $\alpha_{(K,i,i')}(x_1,\dots,x_p)$ holds.''
\end{quotation}
The conjunction $\beta_{\rm ext}:=\bigwedge_{0\leq i\leq i'\leq j}\beta_{K,i,i'}$ states that all $\mathscr{F}{\mid}_i$-compatible assignments on at most $i<i'\leq j$ points extend to $\mathscr{F}{\mid}_{i'}$-compatible assignments on any further $i'-i$ points.

The sentence $\beta_{\rm hom}\wedge \beta_{\rm ext}$ states that the $\mathscr{F}_{\leq k}$-compatible assignments on at most $j$ points is a $(j-1,j)$-strategy.
\end{proof}

The following gives case (1) of Theorem \ref{thm:gap} (and more).
\begin{thm}
Let $\mathbb{A}$ be a template for which the problem $\CSP(\mathbb{A})$ has bounded width $\ell$ and let $\mathscr{F}$ be any finite set of pp-formul{\ae}, and $k$ be an integer larger or equal to $\ell$.
\begin{enumerate}
\item The promise problem $(Y_{(k,\mathscr{F})},N_{\CSP{}})$ for $\mathbb{A}$ is solvable in $\texttt{AC}^0$.  Equivalently, the $(k,\mathscr{F})$-robustly satisfiable instances and the NO instances of $\CSP(\mathbb{A})$ have a separator that is decidable in $\texttt{AC}^0$.
\item If $\CSP(\mathbb{A})$ has strict width $\ell$, then the $(k,\mathscr{F})$-robustly satisfiable instances are themselves decidable in the class $\texttt{AC}^0$.
\end{enumerate}
\end{thm}
\begin{proof}
(1) 
We assume throughout that $\mathbb{B}$ satisfies the $(Y_{(k,\mathscr{F})},N_{\CSP{}})$ promise.\\If $\mathbb{B}$ is $(k,\mathscr{F})$-robustly satisfiable (the YES promise) then by Lemma \ref{lem:jstra} it satisfies the sentence provided by Lemma \ref{lem:FOstrategy}.  Conversely, if $\mathbb{B}$ satisfies the sentence found by Lemma \ref{lem:FOstrategy}, then as the sentence asserts the existence of a $(\ell-1,\ell)$-strategy and $\CSP(\mathbb{A})$ has bounded width $\ell$ implies that $\mathbb{B}$ is a YES instance of $\CSP(\mathbb{A})$, hence does not satisfy the $N_{\CSP}$ promise.  So $\mathbb{B}$ satisfies the $Y_{(k,\mathscr{F})}$ promise as required.

(2) The sentence in Lemma \ref{lem:FOstrategy} describes the $(k,\mathscr{F})$-robustly satisfiable instances.  Certainly it is satisfied by the $(k,\mathscr{F})$-robustly satisfiable instances, by Lemma \ref{lem:jstra}.  Now assume that $\mathbb{B}$ satisfies the sentence.  In particular, $\mathbb{B}$ has a $(k-1,k)$-strategy, so is a YES instance of $\CSP(\mathbb{A})$.  Let $\nu:\mathbb{B}\to\mathbb{A}$ be an $\mathscr{F}$-compatible partial assignment from $k$ elements.  Then $\nu$ is part of the $(\ell-1,\ell)$-strategy.  Because $\CSP(\mathbb{A})$ has strict width, the partial homomorphisms of this strategy extend to homomorphisms: the algorithm presented in Feder and Vardi \cite{fedvar} uses the near unanimity polymorphism to extend strategies to higher and higher values (in particular, with all the original partial homomorphisms continuing to be extended).  In particular $\nu$ extends a full homomorphism, showing that $\mathbb{B}$ is $(k,\mathscr{F})$-robustly satisfiable.
\end{proof}

We mention that a corollary of (1) is that $(Y_{(k,\mathscr{F}),\mathsf{Q}},N_{\CSP{}})$ is also in $\texttt{AC}^0$.  Something similar can be achieved in part (2), though the sentence needs to amended to include the 
$(k,\mathscr{F})_q$-theory of $\mathbb{A}$.

\section{Appendix: proof of Theorem \ref{thm:gap}(2)}
The hardness case (case 2) of Theorem \ref{thm:gap} is only sketched in Section \ref{sec:BW}.  We now give more detail.  As explained in Section \ref{sec:BW}, it suffices to show that $(Y_{(k,\mathscr{F}),\mathsf{Q}},N_{\CSP})$ is $\texttt{Mod}_p(\texttt{L})$-hard for the template  $\mathbb{C}$ for systems of ternary linear equations over an abelian group structure ${\bf C}$ on the universe $C$.

Let $\mathbf{C}=\langle C;\oplus\rangle$ be an abelian group on a finite set $C$ with zero element $0$ and let $g$ be an element of prime order $p$ in $\mathbf{C}$.  Let $\mathbb{C}$ be a finite relational structure on the same set, whose relations include the ternary relations $r:=\{(x,y,z)\mid x\oplus y\ominus z=0\}$, let $\underline{g}$ and $\underline{0}$ denote the unary relations $\{g\}$ and $\{0\}$ on $C$.  Larose and Tesson \cite{lartes} show that $\CSP(\langle C;r,\underline{g},\underline{0}\rangle)$ is $\Mod_p(\texttt{L})$-complete with respect to first order reductions.  Their system either has no solutions or has exactly one solution.  It is relatively easy to transform this into a system in which the solution space is either empty or has dimension $k$ (for any predetermined $k$), and we employ such a method in our argument.  The difficulty is in maintaining the purely locally determined flexibility in this system required by $(k,\mathscr{F})$-robust satisfiability. 

We consider slightly richer language, but which is nevertheless pp-equivalent to $\{r,\underline{g},\underline{0}\}$, which is sufficient by Theorem \ref{thm:ppGAP}.
For $h\in \{0,g\}$, we let  $r_{h}:=\{(x,y,z)\mid x\oplus y\ominus z=h\}$, and $s_h=\{(x,y,z)\mid x\oplus y\oplus z=h\}$, where $\ominus$ denotes subtraction with respect to $\oplus$.  Let $\mathbb{C}$ denote the structure $\langle C;r_0,r_g,s_0,s_g\rangle$, which we now observe is pp-equivalent to $\langle C;r,\underline{g}\rangle$, using equality-free formul{\ae} (so that Theorem \ref{thm:ppGAP} provides a first order reduction).  First note that $r_0$ coincides with $r$.  Next, $\underline{h}$ is definable from $\{r_0,r_g,s_0,s_g\}$ by $\{(x)\mid (x,x,x)\in r_h\}$.  In the other direction, $r_h$ is defined from $\{r,\underline{g},\underline{0}\}$ by way of $\{(x,y,z)\mid \exists u\exists v\ (x,y,u)\in r\And (z,v,u)\And v\in\underline{h}\}$, while $s_h$ is definable by 
\[
\{(x,y,z)\mid \exists u\exists {-}z\exists z_h\, z_h\in\underline{h} \And u\in \underline{0}\And (z,{-}z,u)\in r\And (x,y,{-}z)\in r\}.
\]

\begin{thm}\label{thm:abelian}
For every $k\geq 0$ there exists a finite set of pp-formul{\ae} $\mathscr{F}$ such that the promise problem $(Y_{(k,\mathscr{F})},N_{\CSP})$ is $\Mod_p(\texttt{L})$-hard for $\mathbb{C}$ with respect to first order reductions.
\end{thm}
\begin{proof}
The proof is by induction on $k$, with the base case corresponding to the $\Mod_p(\texttt{L})$-hardness of $\CSP(\mathbb{C})$, which coincides with $(Y_{(0,\varnothing)},N_{\CSP})$, shown to be $\Mod_p(\texttt{L})$-hard in \cite[Theorem 4.1(2)]{lartes}.  Let us assume that we have established $\Mod_p(\texttt{L})$-hardness of $(Y_{({\leq} k,\mathscr{F})},N_{\CSP})$ for $\mathbb{C}$ (for some $\mathscr{F}$).  We provide a first order reduction to $(Y_{(\ell,\mathscr{G})},N_{\CSP})$ for $\mathbb{C}$, with the set $\mathscr{G}$ to be constructed in the course of the proof and with $\ell=\lfloor 3k/2+1\rfloor$.  (Thus at $k=0$ we can achieve $\ell=1$, then $\ell=2$ at the next induction step, then $\ell=4$ and so on.)  For most of the proof we aim higher, at $\ell':=3k+2$, with only the final step dropping a factor of $2$.

The proof will be performed in the language of linear equations on $\mathbf{C}$, with the final equations trivially expressible in terms of $r$ and $s$ (with $+$ and $-$ interpreted as $\oplus$ and $\ominus$).  We are presented with a system of linear equations over $C$, with each equation of the form $x+y\pm z=0$ or $x+y\pm z=g$ (where $\pm$ abbreviates the fact that we have two options, for plus and for minus).  

We first reduce $(Y_{(k,\mathscr{F})},N_{\CSP})$ for this system to $(Y_{(3k+2,\mathscr{G})},N_{\CSP})$ for some systems of 9-ary equations over $C$, and for some finite set of pp-formul{\ae} $\mathscr{G}$; this step is reminiscent of Gottlob's reduction in Section \ref{sec:3k3SAT}, except the specific replacement there does not work here.  We let $X$ denote the set of variables and let $X^+$ denote $\{x_L,x_M,x_R\mid x\in X\}$.  We now replace each variable $x$ with the triple of variables $x_{L}, x_M, x_R$ (left, middle and right) and replace equations $x+y\pm z=h$ (for $h\in\{0,g\}$) by 
\begin{equation}
x_L+x_M+x_R+y_L+y_M+y_R\pm (z_L+z_M+z_R)=h.\tag{$\heartsuit$}\label{eqn:linear}
\end{equation}
This new system is equivalently satisfiable, because we have simply replaced each $x$ by $x_L+ x_M+x_R$.  Thus, if the original system has no solution, then so also the new system has no solution. (We omit the obvious details that this is a first order reduction.)  Now we must show that if the original system is $(k,\mathscr{F})$-robustly satisfiable, then the new system is $(3k+2,\mathscr{G})$-robustly satisfiable, for suitable $\mathscr{G}$.

First observe that the same variable replacement trick can also be made to formul{\ae} in~$\mathscr{F}$: for each formula $\phi(x_1,\dots,x_{k'})$ in $\mathscr{F}$ (for some $k'\leq k$), we let 
\[
\phi_{L,M,R}(x_{1,L},x_{1,M},x_{1,R},\dots,x_{k',L},x_{k',M},x_{k',R})
\] 
be the formula obtained from $\phi$ by replacing each variable by three copies, and each equation (of length $3$) by the length $9$ equation, as in Equation \ref{eqn:linear}.  Note that we also apply this to quantified variables.  We let $\mathscr{G}$ denote the formul{\ae} obtained from $\mathscr{F}$ in this way.

Let us consider a $\mathscr{G}$-compatible partial assignment $\nu$ on some subset $D\subseteq X^+$.  Let us say that a variable $x\in X$ is \emph{determined} by $D$ if all three of $x_L,x_M,x_R$ are in $D$.  We argue that provided at most $k$ variables are determined, then $\nu$ can be extended.  In particular, this is true when $|D|\leq 3k$, but we will also use the more general fact below, and state it for later reference.

{\bf Claim.}  If $E=\{x_1,\dots,x_{k'}\}\subseteq X$ is the set of determined variables, with $k'\leq k$, and $\sigma$ is the $(k,\mathscr{F})$-type of $(x_1,\dots,x_k)$, then $\nu$ extends to a full solution, provided that it preserves $\sigma_{L,M,R}$.
\begin{proof}[Proof of claim]
Assume $\nu$ preserves $\sigma_{L,M,R}$.  It follows that assigning the value $\nu(x_{i,L})\oplus\nu(x_{i,M})\oplus \nu(x_{i,R})$ to $x_i$ (for $i=1,\dots,k'$) produces a map $\nu_X:E\to C$ preserving $\sigma(x_1,\dots,x_{k'})$, hence extends to a full solution $\nu_X$ from $X$.  The values $\nu_X$ on a determined variable $x_i\in E$ are consistent with the values given $x_{i,L},x_{i,M},x_{i,R}$ by $\nu$.  For any non-determined variables, $x_j$, there are at least one of $x_{j,L},x_{j,M},x_{j,R}$ unassigned by $\nu$.  We may extend $\nu$ to all remaining variables in $X^+\backslash D$ by giving the unassigned variables values so that $\nu(x_{j,L})+ \nu(x_{j,M})\pm \nu(x_{j,R})=\nu_X(x_j)$.
\end{proof}

We now reduce the equation length from $9$ to $4$, preserving the robust satisfiability level at $3k+2$ (a further step, reducing from equation length $4$ to $3$ is performed after that).  This is analogous to the argument in Section \ref{sec:3SAT}, and the proof follows a similar theme, albeit for a completely different computational problem.
For each equation $e$ as in \eqref{eqn:linear}, we introduce three new variables $u_{e,L}$, $u_{e,M}$ and $u_{e,R}$, replacing Equation \eqref{eqn:linear} by the quadruple of equations
$x_L+y_L\pm z_L=u_{e,L}$,  $x_M+y_M\pm z_M=u_{e,M}$,  $x_R+y_R\pm z_R=u_{e,R}$ and $u_{e,L}+u_{e,M}+u_{e,R}=h$ (with $\pm$ chosen $+$ or $-$ consistently with the value of $\pm$ in $e$).  Let $X^{+2}$ be the new set of variables.
The new variables simply correspond to a regrouping of \eqref{eqn:linear}, so the equation systems are equivalently satisfiable.  In particular, if the system $\mathbb{B}$ is unsatisfiable, then so is the constructed system.  Now we must show that it is $(3k+2)$-robustly satisfiable, relative to some finite set of pp-formul{\ae}~$\mathscr{H}$.

We refer to variables of the form $x_L,x_M,x_R$ as \emph{lower} variables, and variables of the form $u_{e,L},u_{e,M},u_{e,R}$ as \emph{upper} variables.  For any subset $D\subseteq X^{+2}$, let $D_L^{\downarrow}$ denote the lower variables in $D$ with subscript $L$, and similarly with $M$ and $R$; let $D^\downarrow:=D_L^{\downarrow}\cup D_M^{\downarrow}\cup D_R^{\downarrow}$.  Let $D_L^{\uparrow}$, $D_M^{\uparrow}$,$D_R^{\uparrow}$, denote the corresponding subsets in terms of upper variables and $D^\uparrow:=D_L^{\uparrow}\cup D_M^{\uparrow}\cup D_R^{\uparrow}$.  

This regrouping is easily seen to be a pp-reduction, and the formul{\ae} $\mathscr{H}$ will, as in Section \ref{sec:ppstable} be claw formul{\ae} (the lower variables are essentially the open elements, while the upper variables are the $\exists$-elements).  The argument deviates from that in Section \ref{sec:ppstable}, because we do not have arbitrary robustness to play with (and in this respect is closer to that of Section~\ref{sec:3SAT}).  Also, we use the Claim, so technically call back to the original system of ternary equations in $X$.  The set $\mathscr{H}$ will consist of the claw formul{\ae} of arity $3k+2$ over $\mathscr{G}$ and of bound $3k$.   We will need only those claw formul{\ae} $\exists U \gamma\And \sigma_{L,M,R}$ for which $\sigma$ is a $(k',\mathscr{F}{\mid}_{k'})$-type, and the $3k'$ unquantified variables in the wrist $\sigma_{L,M,R}$ have been made to coincide with some set of $k'$ triples ($L$, $M$ and $R$ copies) of lower level variables from the talon $\gamma$.  

Assume that we have a subset $D\subseteq X^{+2}$ of at most $\ell$ variables and a partial assignment $\nu:D\to C$.  Our goal is to replace $\nu$ with an assignment $\nu'$ on some set $F$ such that any extension of $\nu'$ to a full solution agrees with $\nu$ on $D$, and such that $F^\uparrow$ is empty and $F^\downarrow$ determines at most $\lfloor \ell/3\rfloor$ variables from $X$.  We then show that our compatibility formul{\ae} $\mathscr{H}$ ensure that $\nu'$ is $\mathscr{G}$-compatible, and hence will extend to a full solution.  


At least one of $D_L^{\downarrow}$, $D_M^{\downarrow}$ and $D_R^{\downarrow}$ has size at most $\lfloor (3k+2)/3\rfloor= k$.  Without loss of generality, assume that $|D_R^\downarrow |\leq k$.  Let $F$ denote the set of lower variables consisting of $D^\downarrow$, along with 
\begin{itemize}
\item $x_{i,L}$ if $x_i$ appears in an equation $e$ for which $u_{e,L}$ or $u_{e,R}$ is in $D$ and 
\item $x_{i,M}$ if $x_i$ appears in an equation $e$ for which $u_{e,M}$ or $u_{e,R}$ is in $D$.
\end{itemize}
Observe that $F^\uparrow$ is empty, while $F_R=D_R^\downarrow$, so that $F$ still determines at most $k$ variables from $X$; denote these by $E$ and let $k':=|E|\leq k$.  If we can extend an evaluation of the variables in $F$ to a full solution of our system, then this forces values for variables in $D$.  For $u_{e,L}$ (with $e$ the equation $x_{i_1}+x_{i_2}\pm x_{i_3}=b$) this is because we have now included the variables $x_{i_1,L},x_{i_2,L},x_{i_3,L}$, and the equation 
$x_{i_1,L}+x_{i_2,L}\pm x_{i_3,L}=u_{e,L}$ is in our system.  A dual statement applies to $u_{e,M}$.  For $u_{e,R}$ it is because we have included the variables $x_{i_1,L},x_{i_2,L},x_{i_3,L},x_{i_1,M},x_{i_2,M},x_{i_3,M}$ and the equations $x_{i_1,L}+x_{i_2,L}\pm x_{i_3,L}=u_{e,L}$, $x_{i_1,M}+x_{i_2,M}\pm x_{i_3,M}=u_{e,M}$ and $u_{e,L}+u_{e,M}\pm u_{e,R}=b$ are in our system.  Our goal then is to select a $\mathscr{G}$-compatible assignment for $F$ that forces the previously given values for $D$.  This is enabled by the claw formul{\ae}.

There are at most $3k+2$ equation subscripts $e$ appearing amongst the members of $D^\uparrow$.  Let $\gamma$ denote the conjunction of the equations containing these elements: it is a system of at most $3k+2$ equation families (created in the pp-reduction: replacement of each 9-variable equation by a system of 4 equations of lengths 3 and 4).  Recall the set $E$ of variables $x_i\in X$ determined by $F$ and let $\sigma$ be their $(k',\mathscr{F}{\mid}_{k'})$-type.  The formula $\gamma\wedge \sigma_{L,M,R}$, with suitable quantification, is a claw formula in $\mathscr{H}$ and hence is preserved by $\nu$.  Note that all variables in $F\backslash D$ appear in $\gamma$, thus the value in $C$ of $\nu$ on $D$ and witnesses to quantifiers for elements of $F\backslash D$ provide a definition of a partial assignment $\nu_F$ from $F$.  Because $\gamma$ is true, the values given $F$ by $\nu_F$ are consistent with the values given $D$ by $\nu$.  Because $\sigma_{L,M,R}$ is true, the values given $F$ by $\nu_F$ determine an $\mathscr{F}$-compatible assignment $\nu_X$ from $E\subseteq X$.  This extends to a full assignment from $X$, by assumption, and thus $\nu_F$ also extends to a full solution.  As observed, this must agree with $\nu$, as required.

We now perform a final regrouping, introducing three new variables $v_{e,L}$, $v_{e,M}$ and $v_{e,R}$ for each equation subscript $e$, and replacing  $x_L+y_L\pm z_L=u_{e,L}$ by the equivalent system $x_L+y_L-v_{e,L}=0$ and $u_{e,L}\mp z_L-v_{e,L}=0$, and similarly for the middle and right versions.  The case where $\mp$ is negative in $u_{e,L}\mp z_L-v_{e,L}=0$ can be rewritten as $z_L+v_{e,L}-u_{e,L}=0$.  In this way, all equations are expressed as a sum of three variables with at most one negative variable in each, so directly encode into the relations $\{r_0,r_g,s_0,s_g\}$.  Also, it is trivial that the new system is equivalently satisfiable to the original, as it can be considered as the grouping of $x_L+y_L$ and replacement by $v_{e,L}$.  We can again employ claw formul{\ae} to keep track of compatibilities: we use claw formul{\ae} of arity $\lfloor (3k+2)/2\rfloor$ and of bound $3k+2$.

Suppose we are given an assignment $\nu$ on $\lfloor (3k+2)/2\rfloor$ variables $D$ in this newly created system of ternary equations.  Let $D_\exists$ denote the members of $D$ of the form $v_{e,L}$, $v_{e,M}$ and $v_{e,R}$ for some $e$ and $D'$ the set $D\backslash D_{\exists}$.  
The idea is to replace the given values  of each variable of the form $v_{e,L}\in D_\exists$ by a suitable pair of values for $x_L+y_L$ (and similarly for $v_{e,M}$ and $v_{e,R}$ in $D_\exists$).  Let $E$ denote the union of these pairs (one pair for each element of $D_\exists$)  with~$D'$.  Then $|E|\leq 2|D_\exists|+|D'|\leq 2\lfloor (3k+2)/2\rfloor\leq 3k+2$.  The argument will be complete if we can use the $\mathscr{H}'$-compatibility of $\nu$ to make an $\mathscr{H}$-compatible assignment to the values in $E$ that forces the values $\nu$ gave to $D$.  A claw formula achieves this: take the conjunction of the at-most $|D_\exists|$ many equation families involving the elements of $D_\exists$ with the $(|E|,\mathscr{H}{\mid}_{|E|})$-type of the elements of $E$.  After existentially quantifying those variables not in $D$ we have a claw formula of arity $\lfloor (3k+2)/2\rfloor$ over $\mathscr{H}$ and of bound $3k+2$.  If $\nu$ is compatible with this set of formul{\ae}, then witnesses to the quantified elements of $E$ provide the desired $\mathscr{H}$-compatible assignment.  We have arrived back at the original signature, thus completing the induction step.
\end{proof}

The following is an easy corollary of the claim proved in the proof of Theorem \ref{thm:abelian} and does not require the full power of Theorem \ref{thm:abelian}.
\begin{cor}
Let $d$ be an arbitrary positive integer.  For a finite field $\operatorname{GF}(p^n)$, it is $\Mod_p(\texttt{L})$-complete to distinguish those systems of affine equations of length $3$  having no solutions from those for which the solution space has dimension $d$.
\end{cor}
\begin{proof}
Note after the stated claim in the proof of Theorem \ref{thm:abelian} that if the original system has dimension $d$, then there are $d$ variables $x_1,\dots,x_d$ in $X$ which may be given any value.  Then by the claim, there are $3d+2(|X|-d)$ variables that may be given any value in the new system: namely the three variables each replacing $x_1,\dots,x_d$, along with two of the three variables for each of the remaining $|X|-d$ variables.  These variables remaining for the remainder of the proof, inside an equivalent system of equations, so that the final system has solution space of dimension at least $3d+2(|X|-d)$.
\end{proof}

\bibliographystyle{amsplain}

\end{document}